\documentclass[a4paper,12pt,margin=0.5in]{article}
\usepackage[a4paper, top=1.6in, bottom=1.6in,left=1in,right=1in]{geometry}
\usepackage[utf8]{inputenc}
\usepackage{amsmath,amsthm,amssymb,array,xcolor,multicol,verbatim,mathpazo}
\usepackage{hyperref}
\usepackage{tikz}
\usepackage{subcaption}
\usepackage{cleveref}
\usepackage{amssymb}
\usetikzlibrary{arrows}\usetikzlibrary{shapes.multipart}

\usepackage{forest}
\usepackage[show]{ed}
\usepackage[
    backend=biber,
    style=authoryear,
    natbib=true,
    url=false, 
    doi=true,
    eprint=false
]{biblatex}
\usetikzlibrary{patterns,positioning}
\usepackage{graphicx}
\usepackage{enumitem}
\usepackage{algorithm}
\usepackage[noend]{algpseudocode}

\usetikzlibrary{calc}

\tikzset{axis line style/.style={thin, gray, -stealth}}
\usepackage{adjustbox}
\usepackage{array}

\usepackage{macros}

\newtheorem{lem}{Lemma}
\newtheorem{thm}{Theorem}

\newtheorem{prop}{Proposition}

\newtheorem*{axiom}{Assumption}
\newtheoremstyle{named}{}{}{\itshape}{}{\bfseries}{.}{.5em}{\thmnote{#3}}

\theoremstyle{definition}
\newtheorem{ex}{Example}

\newtheorem{defn}{Definition}

\makeatother
\newcounter{parentnumber}
\usepackage{amsmath}
\usepackage{newtxtext,newtxmath}
\DeclareSymbolFont{CMletters}{OML}{cmm}{m}{it}
\DeclareMathSymbol{v}{\mathord}{CMletters}{`v}

\newtheorem{cor}{Corollary}
\title{Ripple Effects\\\large Robust Norms without Punishment
}
\date{Updated: \today
}
\author{Alistair Barton\thanks{Cambridge University, Email: \href{mailto:ab3315@cam.ac.uk}{ab3315@cam.ac.uk}\\
I would like to thank Joyee Deb for feedback and guidance on this project, I am grateful to Debraj Ray, Hamid Sabourian, Dilip Abreu, Mikhail Safronov, Konstantinos Ioannidis, and Arjada Bardhi for helpful comments.}}
\begin{document}
\maketitle
\begin{abstract}
I propose a novel, tractable model of pro-social norms in large communities with slightly altruistic agents. Agents participate in the norm to influence others to participate in the norm, influencing further agents. A continuum of equilibria sustain the pro-social norm, varying in the distribution of how much agents are influenced by their observations. If agents' effective altruism $\alpha$ is larger than their impatience $1-\delta$, equilibria exist that are robust to a population of bad actors. Greater strategic homogeneity increases the robustness of the norm. Robustness is not improved by increasing the number of observers of each action beyond 1.    
\end{abstract}

\begin{quote}
    ``[E]ach of us has the plague within him; no one, no one on earth, is free from it. And I know, too, that we must keep endless watch on ourselves lest in a careless moment we breathe in somebody's face and fasten the infection on him. What's natural is the microbe. All the rest [...] is a product of the human will, of a vigilance that must never falter.''\footnote{In this quote, the character Jean Tarrou is using `the plague' as a metaphor for his experience with vindictive politics and capital punishment.}\hfill Albert Camus, \textit{The Plague} (\citeyear{camus})
\end{quote}

Generally speaking, would you say that most people can be trusted or that you can't be too careful in dealing with others? Your answer to this question determines if you have \textit{generalized trust}, a major element of your community's social capital, and associated with economic growth (\cite{AC10}), political regulation (\cite{AAC10}), financial development (\cite{GSZ04,GSZ08}), and health and happiness (\cite{CB14}).

But trust cannot exist if it is constantly betrayed: it must be supported by a common norm of moral and trustworthy behaviour. Economic theory often describes these norms as motivated by the threat of material punishment (either by direct enforcement or mediated by reputational damage, in repeated games), but moral behaviour often occurs in the absence of any plausible punishment. Such material motivations are an unconvincing explanation for why people generally: return lost wallets (\cite{CMTZ19,RBA26}), punish antisocial behaviour (\cite{FG01}), tip at restaurants (\cite{A20}), refrain from profitable lies (\cite{G05}), or even make an effort to be considerate, helpful, and polite to total strangers.

In this paper, I propose that even slight social preferences can support pro-social norms through a mechanism of social influence. An agent who deviates from the norm breaks social trust, and may influence others to similarly deviate from the norm in the future, thereby influencing how yet more others will act, propagating the original agent's influence in a `ripple'. Agents value this influence, not because they expect it to affect how they are treated, as in ideas of karma and indirect reciprocity, but because they care (however slightly) about the behaviour of strangers towards \textit{others}.

This social preference can be thought of as altruism: there is some non-zero cost $a$ agents are willing to pay for another agent to obtain one unit of welfare.\footnote{$a=0$ describes an agent unwilling to even lift a finger to benefit another, and may constitute a marker of psychiatric disorders (\cite[p. 748]{psychopath}).} Even if $a$ is small, agents can be incentivized to participate in the norm if their influence is sufficiently large. And for a sufficiently patient agent, it is possible to have a significant `ripple' of influence that still dissipates over time, allowing the norm to be \textit{robust} to a small measure of persistent deviators from the norm. This robustness property is not obtained in prior anonymous models of social norms with large populations (discussed below).

A norm in this paper is an equilibrium where the vast majority of players (`all but $\epsilon$') choose the same action. The model is constructed (and equilibria selected) with a few desirable properties in mind:
\begin{enumerate}
    \item \textit{Anonymity}: agents cannot be targeted for punishment based on their past actions.
    \item \textit{Pro-sociality}: the norm is socially efficient, but participating is against myopic incentives.
    \item \textit{Robustness}: the norm persists if a small measure of agents consistently defy the norm.
    \item \textit{Sustainability}: aggregate behaviour does not change over time.
\end{enumerate}
\citet{SW23} show that symmetric norms satisfying properties 1-3 cannot exist in large self-interested populations under weak asymptotic assumptions. The main contribution of this paper is a highly tractable model of norms capable of satisfying all four properties through the introduction of social preferences.

For an agent to be influenced by others' actions, deviating from the norm must be a best response after observing deviations and contributing to the norm must be the best response after observing contributions. Thus agents' influence must precise compensate for their myopic unwillingness to participate in the norm. This characterizes a continuum of equilibria.

Considering robustness introduces a trade-off: increasing agents' influence increases the incentive to participate in the norm, but increasing the influence of deviations is also to the detriment of robustness. We develop a notion of comparing the robustness of equilibria, by measuring the amount they amplify the deviations of a small, consistently deviating population.

There are two main results: first, slight altruism is sufficient to sustain robust norms as long as agents are sufficiently patient --- it suffices that $\alpha+\delta> 1$, where $\delta$ is the discount factor and $\alpha$ is the myopic willingness to contribute to the norm (associated with altruism $a$), when actions are observed by at least one other agent on average\footnote{This is generally a mild condition, assured if (a) actions are certain to impact someone, and (b) people observe when they are impacted by another's choice.}. Secondly, more robust equilibria are associated more \textit{homogeneous} strategies across the population (under a certain natural ordering of pure strategies).

These results suggest that patience and strategic homogeneity are substitutes for altruism in sustaining robust norms: little altruism is needed if both patience and strategic homogeneity are high, but higher altruism can compensate for a loss in either (this interpretation is informal, as homogeneity is an equilibrium property).

In Section \ref{sec:pub}, I consider how increasing the publicity of actions affects the strength of norms in this model. Contrary to how we might think of norms sustained through punishment, increasing publicity does not necessarily increase the robustness of norms: maximum robustness is attainable when actions are observed by less than one other agent on average.


Throughout most of this paper, we restrict attention to a simple class of Finite Deviation strategies. In Section \ref{sec:strat.gen} we show that very little is lost through this restriction. This is argued through two propositions, the first shows that equilibria within this class are representative of the much broader, natural class of equilibria where agents are more influenced by recent observations than later observations; the second shows that Finite Deviation equilibria form the existence--robustness frontier.


In Section \ref{sec:conc} we discuss interpretations and extensions of the model, as well as relations to ethics and Kantian equilibria.



\subsection*{Related Literature}

This work is strongly influenced by community enforcement models (\cite{K92,E94,H95}) that study repeated prisoner's dilemmas with random matching in anonymous and selfish societies. In these equilibria, deviating is contagious, a single deviation creates an epidemic that soon collapses the norm, incentivizing selfish agents to participate in the norm for their future self's sake. 

However these equilibria are not robust to a population of consistent deviators.\footnote{\citet{E94} discusses how the equilibria can be made robust to trembling hands, by randomly `restarting' the game. This is distinct from robustness to consistent deviators who are constantly collapsing the norm. \citeauthor{E94} observes
\begin{quote}
    ``If one player were `crazy' and always played $D$ (or was simply was unaware which equilibrium was being played) [...] contagion strategies would not support cooperation. In large populations, the assumption that all players are rational and know their opponents' strategies may be both very important to the conclusions and fairly implausible.''
\end{quote}} As previously mentioned \citet{SW23} show that (under certain assumptions that are weak without social preferences) the robust norms in large anonymous societies must be myopic. The introduction of slight altruism escapes this result, by effectively creating an externality that scales with the population.

Our relaxation of selfishness is justified by the observations of \citet{Andreoni.paradox,A90} that people's revealed preference for altruism is not crowded out by being part of a large society with other contributing agents.\footnote{That such social considerations feature in agents' preferences was originally proposed by \citet{Edgeworth} under the name \textit{effective sympathy}.} The proposed model of `warm glow' altruism is one possible foundation for the preference we use (derived in Section \ref{sec:altruism.FAQ}). 
\citet{EGK07} exploit a similar scalability property of social preferences as an explanation for high election turnout, despite the improbability of an individual voter being pivotal. 

A broad literature studies norms in repeated games without anonymity, considering additional information structures capable of sustaining norms. A sample of such papers include \citet{OFP95,T10,HM18,BT18,SW21}. \citet{T10} notably develops a notion of \textit{independent and indifferent equilibria}, a property also satisfied by our Finite Deviation equilibria, and related to the \textit{belief free equilibria} of \citet{EHO05}. In Section \ref{sec:mech.ag}, we provide a `mechanism agnostic' interpretation of our equilibria, consistent with \citeauthor{B26}'s (\citeyear{B26}) discussion of these techniques.

Our conception of cascading influence is empirically rooted in the social contagion literature. \citet{FC10,bond2012experiment,FWE15,SMSW23} are a sample that studies contagions of specifically pro-/anti-social behaviour. On an individual level, the contagious nature of pro-social behaviour is also consistent with literatures on pro-social conformity and conditional cooperation (\cite{Frey04,B05,SC09,nook2016}).

\section{Model}

We consider a society with a unit-measure continuum of agents $\mathbb{I}$ making decisions in integer time $\mathbb{N}$ and discounting the future at rate $\delta$. Each period agents decide whether to contribute $C$, or deviate $D$; the set of actions is $A:=\{C,D\}$. Contributing has net cost $c$ to the individual, but produces $c+g$ units of public welfare, for a net gain of $g$ units to public welfare. Agents are slightly altruistic, willing to pay $a$ per marginal unit increase to public welfare. We normalize $c=1$.

After each period, every agent's action is observed by a number $\beta$ of anonymous agents (drawn from an atomless distribution on $\mathbb{I}$), and every agent observes the actions of $\beta$ other agents. For non-integer $\beta$, the number of observed agents is drawn from a distribution with mean $\beta$ and bounded support, independently for each agent.



\subsubsection*{Preferences}\label{sec:modelpref}
An agent's belief about the outcome of the game can be represented by a vector $p\in[0,1]^{\mathbb{I}\times\mathbb{N}}$ where $p_{j,\tau}$ be the probability agent $j$ contributes in period $\tau$. Suppose $p^C$ is the path if agent $i$ contributes in period $t$ while $p^D$ is the path if they deviate. They prefer to contribute if
    \begin{equation}\label{eq:pref}
        \sum_{\tau\ge t} \delta^{\tau-t} \Big(-(p^C_{i,\tau}-p^D_{i,\tau})+\alpha\sum_{j\in\mathbb{I}}(p^C_{j,\tau}-p^D_{j,\tau})\Big)\ge 0,
    \end{equation}
where $\alpha:=\frac{a}{1-a}g$ is the effective altruism of agents associated with the action $C$. Eq. \ref{eq:pref} describes an agent who places an equal value of $\alpha$ on any agent's contribution, and is willing to privately contribute if their contribution is subsidized at a rate $\overline{\delta}_1:=1-\alpha$ (which we call the \textbf{altruism discount}, to contrast with selfish agents who would only contribute if subsidized the full cost $1$). We assume $0<\alpha<1$.


We refer to the second term as the (altruistic) \textbf{externality} associated with contributing/deviating. It describes how an agent weighs the public welfare difference between the two paths. Its summation has an uncountable index --- mathematically, the sum diverges unless the summand is zero on a co-countable subset of $\mathbb{I}$. This does not pose a problem as agents are only capable of influencing a countable number of agents in our society --- for our purposes, the sum could be restricted to the countable set of actions that agent $i$ can influence.\footnote{For similar reasons, the incompleteness of the preference when the sum diverges will not be relevant to our analysis. Indeed, the incompleteness is a positive feature as we avoid taking a stance on how agents resolve moral paradoxes.

 The key role of the continuum population assumption that introduces this complexity is to ensure that an agent's observations are independent of the strategies of their observers, a property we refer to as \textit{observation independence} and discuss in more detail in Section \ref{sec:cont.pop}.}

This preference, and variations upon it, are further discussed in Section \ref{sec:altruism.FAQ}.

\subsubsection*{Strategies} 
To describe the sustainability of equilibria, we view strategies as automata (first developed by \cite{R86}). Recall that an agent's observation consists of a number of contributions and a number of deviations, which can be described as an element $(o_C,o_D)\in\mathscr{O}:=\mathbb{N}_0^2$ given by the number of contributions $o_C$ and deviations $o_D$ they observe. 

A strategy space is a tuple $\langle \Sigma,\sigma^+,\hat{a} \rangle$ composed of a (possibly infinite) measurable space $\Sigma$; a transition function $\sigma^+: A\times\mathscr{O}\times\Sigma\rightarrow \Sigma$ dictating the state that the agent updates to for the next period based on their action, their observations, and their current state; and an action choice for every state $\hat{a}:\Sigma\rightarrow  A$.

A pure strategy is an initial state in this network  $\sigma^0\in\Sigma$,\footnote{\citet{R86} uses the tuple $\langle \Sigma,\sigma^+,\hat{a},\sigma^0 \rangle$ to describe a strategy. We fix the universe of pure strategies $\langle \Sigma,\sigma^+,\hat{a} \rangle$, enabling the shorthand of describing a strategy by an initial state in this universe.} with future actions decided by how the agent's state is updated by $\sigma^+$. A \textbf{population strategy} is a measure $\mu\in\Delta\Sigma$.

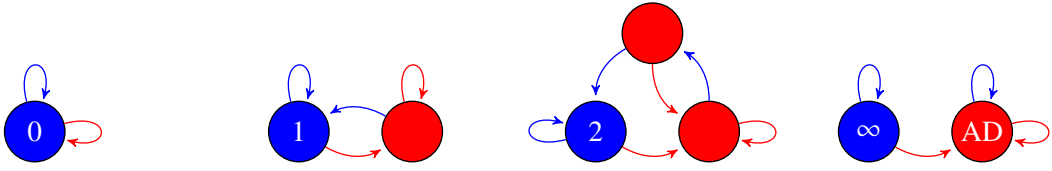
\begin{figure}
    \centering
    \begin{subfigure}{0.24\textwidth}\centering
    \begin{tikzpicture}[->, >= stealth', shorten >=2pt, line width=0.5pt, node distance=2cm, inner sep=0pt, minimum size=8mm,
    C/.style={circle , draw, fill=blue, text =white},
    D/.style={circle , draw, fill=red,text =white},
    Carrow/.style={blue},
    Darrow/.style={red}]
        \node [C] (c) at (0,0) {\small 0};
        \path[Carrow] (c) edge [ loop above ] (c) ;
        \path[Darrow] (c) edge [ loop right ] (c) ;
    \end{tikzpicture}
    \end{subfigure}
    \begin{subfigure}{0.24\textwidth}\centering
    \begin{tikzpicture}[->, >= stealth', shorten >=2pt, line width=0.5pt, node distance=2cm, inner sep=0pt, minimum size=8mm,
    C/.style={circle , draw, fill=blue, text =white},
    D/.style={circle , draw, fill=red,text =white},
    Carrow/.style={blue},
    Darrow/.style={red}]
        \node [C] (c) at (0,0) {\small 1};
        \node [D] (d) at (1.5,0) {};
        \path [Carrow] (c) edge [ loop above ] (c);
        \path [Darrow] (d) edge [ loop above ] (d);
        \path [Darrow] (c) edge [ bend right ] (d);
        \path [Carrow] (d) edge [ bend right ] (c);
    \end{tikzpicture}
    \end{subfigure}
    \begin{subfigure}{0.24\textwidth}\centering
    \begin{tikzpicture}[->, >= stealth', shorten >=2pt, line width=0.5pt, node distance=2cm, inner sep=0pt, minimum size=8mm,
    C/.style={circle , draw, fill=blue, text =white},
    D/.style={circle , draw, fill=red,text =white},
    Carrow/.style={blue},
    Darrow/.style={red}]
        \node [C] (c) at (0,0) {\small 2};
        \node [D] (d1) at (1.5,0) {};
        \node [D] (d2) at (0.75,1.3) {};
        \path [Carrow] (c) edge [ loop left ] (c);
        \path [Darrow] (d1) edge [ loop right ] (d1);
        \path [Darrow] (c) edge [ bend right ] (d1);
        \path [Carrow] (d1) edge [ bend right ] (d2);
        \path [Carrow] (d2) edge [ bend right ] (c);
        \path [Darrow] (d2) edge [ bend right ] (d1);
    \end{tikzpicture}
    \end{subfigure}
    \begin{subfigure}{0.24\textwidth}\centering
    \begin{tikzpicture}[->, >= stealth', shorten >=2pt, line width=0.5pt, node distance=2cm, inner sep=0pt, minimum size=8mm,
    C/.style={circle , draw, fill=blue, text =white},
    D/.style={circle , draw, fill=red,text =white},
    Carrow/.style={blue},
    Darrow/.style={red}]
        \node [D] (d) at (1.5,0) {{\small AD}};
        \path[Carrow] (d) edge [ loop above ] (d) ;
        \path[Darrow] (d) edge [ loop right ] (d) ;
        \node [C] (c) at (0,0) {\small$\infty$ };
        \path[Carrow] (c) edge [ loop above ] (c) ;
        \path[Darrow] (c) edge [ bend right ] (d) ;
    \end{tikzpicture}
    \end{subfigure}
    \caption{A possible strategy space $\langle \Sigma,\sigma^+,a\rangle$ --- nodes are states, coloured according to $a$ (blue indicates $a(\sigma)=C$, while red is $a(\sigma)=D$). Arrows describe $\sigma^+$, blue indicates the transition after observing no deviations, red is after observing one or more deviations. $0,1,2$ denote the respective Finite Deviation contributing states, linked to their associated (unlabelled) deviating states, $0$ is also known as Always Contribute, $\infty$ is Grim Trigger, while AD is Always Deviate.}
    \label{fig:stratex}
\end{figure}

One strategy space is illustrated in Figure \ref{fig:stratex}. The Always Deviate strategy is a state $\sigma_\text{AD}$ that deviates and remains in this state no matter what. The Grim Trigger strategy is identified with a state $\sigma_\infty$ that contributes, and remains in this state until observing a deviation where it switches to $\sigma_\text{AD}$.\footnote{This is subtly different from the typical Grim Trigger strategy which also prescribes that the agent deviate if \emph{they} have previously deviated, even if they have not observed anyone else deviate. This distinction does not have material effects on our results, but is merely for consistency with $n$-Deviation strategies.} The $n$-Deviation strategy (denoted $n$ or $\sigma_n$) contributes iff no deviation has been observed in the previous $n$ periods, and is associated with $n$ deviating states that are incrementally cycled through (if no deviations are observed) until arriving at the contributing state. 

For notational convenience, we denote Grim Trigger as $\infty$ or $\sigma_\infty$, and Always Contribute as $0$ or $\sigma_0$. We will assume that the strategy space includes the \textbf{Finite Deviation strategy space}, consisting of the \textbf{Finite Deviation strategies} --- the set of $n$-Deviation strategies represented by $\mathbb{N}_0$ --- as well as the deviating states they connect to, and Always Deviate. In Section \ref{sec:strat.gen} we show that restricting attention to Finite Deviation strategies is with little loss of generality.

We are interested in sustainable norms, consisting of a population strategy that does not evolve over time:
\begin{defn}
A population strategy $\mu$ is \textbf{stationary} if its distribution over states is constant over time when interacting with itself. Formally, for all measurable $\Sigma'\subseteq \Sigma$
\begin{equation}\label{eq:stat.strat}
\mu[\Sigma']= \sum_{o_C\ge 0}\mu\Big\{\sigma\in\Sigma;\sigma^+\big(a(\sigma),(o_C,\beta-o_C)\rvert\sigma\big)\in\Sigma'\Big\}\P{o_C}
\end{equation}
where the number of contributions observed $o_C$ is distributed according to the binomial distribution $o_C\sim B\big(\beta,\mu[\hat{a}^{-1}(C)]\big)$.\footnote{For non-integer $\beta$, take the expection of eq. \ref{eq:stat.strat} over the number of total observations.}
\end{defn} 
A pure strategy is stationary iff the initial state is unchanged when interacting with itself. All Finite-Deviation strategies are stationary.

A population strategy is stationary if individuals transition between states in ways that keep the aggregate distribution constant --- for example, if $\beta=1$, then any distribution that puts positive weight on both the 1-Deviation strategy and its associated deviating state is stationary: the measure of contributors in one period is then the probability of observing a contribution  the previous period, equal to the measure of contributors that period. However, if we add some positive weight to the Always Deviate state, the resulting population strategy is no longer stationary, as now the measure of agents contributing will always be smaller than the measure of agents that observed a contribution.

Our solution concept is the following:
\begin{defn}
    A \textbf{perfect stationary equilibrium} (\textbf{PSE}) is a stationary strategy that is a best response to itself after any sequence of observations, according to the preference given by eq. \ref{eq:pref}.
\end{defn}

\section{Full Participation Equilibria}\label{sec:exist}
Our first goal is to characterize PSEs where agents contribute on path, ie. $\mu[\hat{a}^{-1}(C)]=1$, which we will call \textbf{full-participation} equilibria. We begin with two examples, demonstrating our technique, before moving to the more general characterization.

\subsection{Example Equilibria}
\paragraph{One-Deviation Equilibrium} Suppose agents randomize between the 1-Deviation strategy (w.p. $\mu_1$) and Always Contribute (w.p. $1-\mu_1$). Trivially this is a stationary strategy. The expected externality incurred by deviating can be calculated recursively, as a single deviation directly results in a loss of $\alpha$, and further leads to $\beta\mu_1$ other agents deviating in the next period (on average).  Thus the expected externality $\Delta$ created by contributing rather than deviating solves
\begin{equation}\label{eq:1TFT.Delta}
    \Delta=\alpha+\delta\beta\mu_1 \Delta 
\end{equation}
For this population strategy, each deviation directly creates $\beta\mu_1$ other deviations --- this is the \textbf{reproduction rate} $R(\mu)$ of deviation in this strategy. Since agents discount the future, they care about the \textbf{discounted reproduction rate} $D(\mu)=\delta\beta\mu_1$ of deviation. These are two measures of agents' equilibrium influence.

For this strategy to be a PSE, agents who observed no deviation in the previous period must be willing to contribute; while those who observed a deviation and are in the 1-Deviation state must be willing to deviate. Since their preference in both scenarios is the same, they must be indifferent. This occurs when the externality $\Delta$ created by contributing is exactly equal to the private cost of contributing. Rearranging,
\begin{equation}
    \label{eq:DDR}
    D(\mu)=1-\alpha=:\overline{\delta}_1.
\end{equation}
The right side is the excess cost of contributing beyond the direct externality it produces, which must then be offset by an agent's influence $D(\mu)$ to create indifference. The lower the altruism discount $\overline{\delta}_1$, the less influence is necessary in equilibrium.

If we tune $\mu_1$ to get equality in eq. \ref{eq:DDR}, we obtain a PSE. This is possible whenever $\delta\ge \overline{\delta}_1/\beta$ (and certainly when $\delta\ge 1/\beta$). Such an equilibrium involves $\mu_1=\overline{\delta}_1/\delta\beta$ agents in the 1-Deviation state and the rest always contributing. 

Note that this equilibrium may exist even if the direct externality $\alpha$ is very small, as long as $1-\delta\beta$ is smaller.

\paragraph{Grim Trigger Equilibrium} Suppose agents randomize between Grim Trigger (w.p. $\mu_\infty$) and Always Contribute (w.p. $1-\mu_\infty$). Trivially this is a stationary strategy. Now a single deviation leads to the possibility of the opponent deviating forever after, leading to the recursive equation for the externality
\begin{equation}\label{eq:infTFT}
\Delta=\alpha+\beta\mu_\infty\sum_{n=1}^\infty \delta^n \Delta=\alpha+\beta\mu_\infty\frac{\delta}{1-\delta}\Delta
\end{equation}
The reproduction rate of deviation in this equilibrium is infinite, as a single deviation expects to directly induce infinitely more\footnote{Note that the reproduction rate is not the rate at which deviation increases over time --- but the rate at which deviation grows as the degree of separation from the original deviator increases.}; however, by comparison with eq. \ref{eq:1TFT.Delta}, the discounted reproduction rate is $D(\mu)=\beta\mu_\infty\frac{\delta}{1-\delta}$. Note that if $D(\mu)\ge 1$, the recursive equation breaks down and the externality $\Delta$ will be infinite.

By tuning $\mu_\infty$ to solve $D(\mu)=\overline{\delta}_1$ we obtain an equilibrium whenever 
$$
\delta\ge \tfrac{\overline{\delta}_1}{\overline{\delta}_1+\beta}=:\overline{\delta}_\infty.
$$
Note that this lower bound $\overline{\delta}_\infty$ is at most $\tfrac{1}{1+\beta}$. This is because the pure Grim Trigger strategy increases the number of deviations by a factor of $1+\beta$ every period, thus the total externality can diverge to infinity as long as $\delta(1+\beta)\ge 1$. While grim trigger can impose strong incentives, this equilibrium is intuitively fragile in a way that will be formalized in Section \ref{sec:robust}.\footnote{\label{fn:interpret}The reader may be skeptical if fragile equilibria have \textit{any} validity in finite population settings. Indeed observation independence fails severely in finite populations where many agents play Grim Trigger --- seeing a deviation now means many agents will be deviating later. However, if the society is very large, these deviations will take a long time to spread, and the time at which agents expect to see meaningful deviation rates may be sufficiently distant to be of little strategic relevance. As such, observation independence may still approximately hold over the time-horizon that agents deem relevant.} 

While we have only explored the discounted reproduction rate for a couple of examples, our analysis generalizes:
\begin{thm}\label{thm:GT}
A stationary population strategy $\mu$ is a PSE iff it solves eq. \ref{eq:DDR}. Full-participation PSE exist if $\delta\ge\frac{1}{1+\beta}$.
\end{thm}
This highlights a stark discontinuity at $\alpha= 0$ where Always Deviate is the unique PSE.

\subsection{General Finite-Deviation Equilibria} 
We represent a Finite Deviation strategy $\mu$ as a distribution over $\mathbb{N}_0$, where $\mu_n$ is the probability of being in the $n$-Deviation state (defined to be contributing). Such a distribution is stationary as previously observed.

A pure $n$-Deviation strategy has discounted reproduction rate
$$
D(n)=\beta(\delta+\cdots +\delta^n)=\beta\tfrac{\delta}{1-\delta}(1-\delta^n).
$$
A population Finite-Deviation strategy $\mu\in \Delta\mathbb{N}_0$ has discounted reproduction rate
\begin{equation}\label{eq:genexist}
    D(\mu)
    =\beta\tfrac{\delta}{1-\delta}\E[n\sim\mu]{1-\delta^n}.
\end{equation}
When $\delta>\overline{\delta}_\infty$ there are a continuum of Finite Deviation equilibria $\mu$ solving $D(\mu)=\overline{\delta}_1$. We parametrize this equilibrium space (hyperplane) by the \textbf{reactivity distribution} $\mu^+\in\Delta\{1,2,\dots\}$, the probability that an agent is in the $n$-Deviation state conditional on not being in the 0-Deviation state (ie. the describes the population strategy of agents who react to a deviation). 

Appropriately randomizing between 0-Deviation and $\mu^+$ then obtains an equilibrium when
\begin{equation}\label{eq:overest}
\overline{\delta}_1\le D(\mu^+)=\tfrac{\beta\delta}{1-\delta}\E[n\sim\mu^+]{1-\delta^n}.
\end{equation}
Note $n\mapsto 1-\delta^n$ is increasing and concave. Recall that a distribution $\mu$ is larger than another $\mu'$ according to second-order stochastic dominance (\textbf{SOSD}) if $\E[n\sim\mu]{f(n)}\ge \E[n\sim \mu']{f(n)}$ for any increasing, concave function $f$; $\mu$ is larger according to first-order stochastic dominance (\textbf{FOSD}) if $\E[n\sim\mu]{f(n)}\ge \E[n\sim \mu']{f(n)}$ for any increasing function $f$. This obtains the following proposition:
\begin{prop}\label{prop:exist}
If an equilibrium exists for a reactivity distribution $\mu^+$, then it exists for any SOSD-larger reactivity distribution $\hat{\mu}^+$. 

No two Finite Deviation PSEs are SOSD-comparable in $\Delta\mathbb{N}_0$ (\textit{a fortiori}, they are not FOSD-comparable).
\end{prop}
The latter statement is a consequence of the equilibrium characterization eq. \ref{eq:DDR}. The former statement follows from the `overestimation' of eq. \ref{eq:overest}, and implies that whenever a 1-Deviation equilibrium exists (ie. $\delta\ge \overline{\delta}_1/\beta$), there exists an equilibrium for \textit{any} reactivity distribution $\mu^+$.

This result is intimately connected with axiomatic properties of agents' preferences. That $D(\cdot)$ is the expectation of an increasing concave function follow from agents being present-biased: the difference between the $(n-1)$-Deviation strategy and the $n$-Deviation strategy is a deviation $n$-periods in the future. As $n$ increases the weight placed on this marginal deviation decreases while remaining positive. This connection is formally established in the Online Appendix.

A mean-preserving contraction can also be thought of as moving deviations from the end of the deviation phase of high-$n$ strategies in $\mu$ to the end of the deviation phase of low-$n$ strategies, effectively moving the deviations earlier in time thus increasing. 

\section{Behavioural Robustness \& Stability}\label{sec:robust}

One reason the Grim Trigger equilibrium is unappealing is that the stationarity of the equilibrium critically relies on there being no deviators. As soon as there is a positive measure of deviators, agents will start transitioning out of the Grim Trigger state, never to return.

We formalize this concern through the two notions of stability and robustness. Define $\lambda_t$ to be the probability with which deviations are observed in period $t$. In a PSE $\lambda_t\equiv\mu[a^{-1}(D)]$.
\begin{defn}[Dynamic Stability]
    A stationary strategy $\mu$ with deviation rate $\lambda_\ast$ is \textbf{locally stable} if, for sufficiently small $\epsilon$, after any finite history with deviation rates $\lambda_{-t}$\footnote{We use negative indices to highlight the exogenous nature of these deviation rates, distinct from the strategy-induced deviations that occur in `positive' time.} satisfying $|\lambda_{-t}-\lambda_\ast|<\epsilon$, the population strategy will converge back to $\mu$ with the deviation rate in any period satisfying $|\lambda_t-\lambda_\ast|<\epsilon$.

    It is \textbf{(globally) stable} if this holds for any $\epsilon>0$. 
\end{defn}
Fluctuations in the deviation rate will influence more agents to deviate, generally increasing the deviation rate. Stability just says that the strategy will course-correct over time without the deviation rate going further away from $\lambda_\ast$.\footnote{\citet{K92,E94} use a weaker notion of \textit{global stability}, allowing for the deviation rate to temporarily explode before the population strategy eventually returns to $\mu$ (in \citeauthor{E94}'s case through a coordination device). Their community enforcement equilibria do not satisfy this stronger property.}

Stability is a purely dynamic property --- it assumes that agents do not adjust their strategy to these fluctuations which may affect their best response. The following notion of robustness is strategic, saying that agents should be able to adjust their strategy to a small subpopulation of always deviators, without significantly disrupting the equilibrium:
\begin{defn}[Behavioural Robustness]
    A full-participation PSE $\mu$ is \textbf{(behaviourally) robust} if there exists a sequence of PSE $\tilde{\mu}$ where agents Always Deviate with positive probability ($\tilde{\mu}[\sigma_\text{AD}]>0$) with $\tilde{\mu}\rightarrow \mu$ (under the weak topology).

    For two behaviourally robust equilibria $\mu,\mu'$, we say $\mu$ is \textbf{more robust} than $\mu'$ if 
\begin{equation}\label{eq:rob}
    \min_{\tilde{\mu}\rightarrow \mu}\lim \frac{\lambda_{\tilde{\mu}}}{\tilde{\mu}[\sigma_\text{AD}]}<\min_{\tilde{\mu}'\rightarrow \mu'}\lim\frac{\lambda_{\tilde{\mu}'}}{\tilde{\mu}'[\sigma_\text{AD}]},
\end{equation}
where the minimum is over approximating sequences of PSEs $\tilde\mu\rightarrow\mu,\tilde\mu'\rightarrow\mu'$.\footnote{The minimums in eq. \ref{eq:rob} should be thought of as a suitable refinement of the topology on $\Delta \Sigma$. For example, in Finite Deviation space the minimum will be obtained for any sequence satisfying $\sum n|\tilde\mu_n-\mu_n|\rightarrow 0$.}
\end{defn}
In essence, we are looking for full-participation equilibria that are also valid descriptions of `almost full-participation' stationary equilibria. 
When we add a measure of Always Deviators to an equilibrium, their deviation is amplified through their influence. An equilibrium is more robust than another if it amplifies deviation by a smaller multiplier.

The key factor in determining the stability/robustness of an equilibrium is the aforementioned reproduction rate, defined as how many deviations a single deviation causes in those that directly observe it: if $\mu$ is a Finite Deviation strategy, the \textbf{reproduction rate} is $R(\mu)=\beta\E[n\sim\mu]{n}$.
\begin{lem}\label{lem:rob}
    For Finite Deviation strategies, local and global stability are equivalent, and both are implied by robustness. A necessary condition for stability and robustness is $R(\mu)\le 1$; a sufficient condition for both is $R(\mu)<1$.

    For two behaviourally robust PSE $\mu,\mu'$, the strategy $\mu$ is strictly more robust than $\mu'$ if
    $$
    R(\mu)<R(\mu').
    $$
\end{lem}
Equivalence of local and global stability holds as Finite Deviation strategies have maximum influence if there is full participation in the norm. With higher deviation rates, some agents may be in deviating states for high-$n$ strategies, reducing the influence a deviation has on them. This same effect makes robustness more difficult to achieve since it involves a slightly higher deviation rate (although the concepts only differ in the edge case $R(\mu)=1$). Since robustness is the strongest criterion for Finite Deviation equilibria, we will often refer solely to equilibria as robust, leaving stability implied.

The last statement says that more robust equilibria are intuitively those that induce fewer deviations. This contrasts with the criterion for the existence of equilibria in Prop. \ref{prop:exist}, which favours more reactive strategies, highlighting the tension between robustness and existence. 

 We defer the proof to Appendix \ref{app:proof}, which considers the result for more general strategies and analyzes the edge case $R(\mu)=1$.\footnote{In this case, stability and robustness depends on other details of the equilibrium, for example robustness may depend on $\beta$ or the structure of our strategy space (allowing equilibria to be approximated from different directions). In Proposition \ref{prop:rob.edge} we specify how Finite Deviation strategy space is sufficient to determine robustness under (additional) natural conditions on the approximating sequence.}

The discounted reproduction rate $D(\cdot)$ and reproduction rate $R(\cdot)$ are both increasing in $\mu$ under the FOSD order, but the present-bias of $D(\cdot)$ means that it is increasing under mean-preserving contractions, while $R(\cdot)$ is a long-run statistic that is preserved by mean-preserving contractions. This suggests that deviations should be front-loaded as much as possible to increase their salience to agents at minimal cost to robustness. 

This front-loading process can be described by the convex-SOSD order (\textbf{cSOSD}): we say $\mu$ is cSOSD-larger than $\mu'$, if $\E[n\sim\mu]{g(n)}\ge \E[n\sim\mu']{g(n)}$ for all increasing, \textit{convex} $g$. A cSOSD-smaller strategy will thus have a smaller reproduction rate $R$, but potentially a larger discounted reproduction rate $D$.
\begin{cor}    
    If $\mu,\mu'$ are robust Finite Deviation PSE such that $\mu'$ is cSOSD-smaller than $\mu$, then $\mu'$ is more robust than $\mu$. Such PSE exist iff there is a non-trivial mean-preserving contraction $\hat{\mu}\neq\mu$ of $\mu$.
\end{cor}
The first statement is immediate from Lemma \ref{lem:rob}, to find such a PSE for the second part we could take the PSE with reactivity distribution $\hat{\mu}^+$. 

This result describes the specific sense in which strategic homogeneity increases robustness --- Proposition \ref{prop:exist} shows that PSEs are never FOSD-comparable, thus it is specifically the non-FOSD comparisons of the cSOSD order that are relevant (ie. those that involve a mean-preserving contraction). 


Applying the previous corollary, the maximally robust equilibria will be supported on two consecutive Finite Deviation strategies. We give this equilibrium a name:
\begin{defn}\label{def:concentrate}
    A Finite Deviation strategy $\mu$ is \textbf{concentrated} if $\supp{\mu}\subseteq\{n,n+1\}$ for some $n\in\mathbb{N}_0$.
\end{defn}
A special case is the 1-Deviation equilibrium which has support $\{0,1\}$. Concentrated strategies are FOSD ordered, and thus there will be a unique concentrated equilibrium, which exists whenever a non-Grim Trigger equilibrium exists (ie. $\delta>\overline{\delta}_\infty$).

The previous proposition says that Finite Deviation equilibria with more concentrated distributions will have lower reproduction rates. 
\begin{cor}
    If any robust PSE $\mu$ exists, then the concentrated equilibrium is robust, and is the most robust equilibrium (uniquely so among Finite Deviation PSEs).
\end{cor}
To this point we have restricted attention to Finite Deviation PSEs, this will be extended to general PSEs via Proposition \ref{prop:mono.suff}, which shows that Finite Deviation PSEs are the most robust PSEs.

For the concentrated equilibrium $\mu$ with support $\{n,n+1\}$ to be robust, it is necessary that $n\beta<R(\mu)\le 1$. Thus if $\beta\ge 1$ and there exists a robust equilibrium, then the 1-Deviation equilibrium is robust.

\section{Publicity Effects}\label{sec:pub}
The publicity $\beta$ of actions varies greatly across contexts. In urban settings, the actions made by individuals in public are observed by many more people than actions taken in more rural settings.

This increased publicity enables agents to have more influence, as we have seen, additional influence has the ability to increase the incentive to contribute, as well as decrease the robustness of equilibria. Theorem \ref{thm:GT} showed that increasing $\beta$ intuitively makes it easier to support fragile equilibria; we will see that for the most robust equilibrium, these two effects of increased publicity largely cancel each other out.

Let $\overline{\beta}_1:=\overline{\delta}_1/\delta$ be the minimum publicity for the 1-Deviation equilibrium to exist.

\begin{prop}[Publicity Effects]\label{prop:pub.rob}
    If there exists a robust equilibrium at publicity $\beta$, then one exists for larger publicities $\beta'>\beta$. If $\beta,\beta'>\overline{\beta}_1$ then a robust equilibrium exists at $\beta$ iff it exists at $\beta'$.

    The robustness of the concentrated equilibrium is strictly increasing in $\beta$ until $\beta=\overline{\beta}_1$ and constant thereafter.
\end{prop}
This says that publicity strengthens norms, but only until the threshold $\overline{\beta}_1$: if a robust equilibrium exists for any publicity, then the 1-Deviation equilibrium with publicity $\overline{\beta}_1$ is robust. Such an equilibrium would have reproduction rate $\overline{\beta}_1$, immediately (modulo the edge case $\overline{\beta}_1=1$, treated in Proposition \ref{prop:rob.edge}) leading to the following result:
\begin{thm}\label{thm:pub.rob}
    If an equilibrium is robust for some $\beta$, then $\overline{\beta}_1\le 1$ (equivalently $\delta\ge \overline{\delta}_1$). 
    
    This is a sufficient condition for a  robust equilibrium to exist if $\beta>\overline{\beta}_1$, in which case the most robust equilibrium $\mu$ will have reproduction rate $R(\mu)=\overline{\beta}_1.$
\end{thm}
Robust norms can exist even if actions are relatively private, and increasing the publicity beyond 1 has no effect on the maximum robustness. This also says that a necessary condition for robust PSE to exist is $\delta\ge \overline{\delta}_1$; that is, agents' effective altruism is larger than their impatience: $\alpha\ge 1-\delta$. This becomes a sufficient condition if $\beta>\overline{\beta}_1$ (this does not rule out robust equilibria for lower $\beta$).

As a back of the envelope calculation, suppose agents have discount rate $\delta=0.9$ and pure altruism $a=0.25$,\footnote{\citet{Dictator2} employ a dictator game experiment to estimate CES utility functions for altruism. On the margin, they estimate $a\approx 0.32$ for subjects they classify as `weakly selfish' (fairness concerns can shift this number up or down).} In this case Ripple Effects can sustain robust norms whose social benefit exceeds their cost by $g>30\%$ (assuming $\beta\ge 1$).
\begin{proof}[Proof of Proposition \ref{prop:pub.rob}]
    Conisder the concentrated equilibrium $\mu$ at publicity $\beta$, and consider a model with publicity $\beta'>\beta$. Then the strategy
    $$
    (1-\tfrac{\beta}{\beta'})\sigma_0\oplus \tfrac{\beta}{\beta'}\mu
    $$
    has the same reproduction rate and discounted reproduction rate as $\mu$. If $\mu$ is not the 1-Deviation equilibrium, then this is not the concentrated equilibrium, and thus the robustness can be improved upon. 
    
    On the other hand, if $\mu$ is the 1-Deviation equilibrium, the result is a concentrated equilibrium and robustness cannot be improved upon.
\end{proof}
While increasing publicity never decreases the robustness of the most robust equilibrium, it may have a negative effect depending on how it affects equilibrium selection. 

As a simple example, suppose that the publicity increases from $\beta_0$ to $\beta_1$, but the shock is only revealed to a measure $1-\gamma$ of agents, the ignorant measure $\gamma$ believe publicity is still $\beta_0$. Assuming that strategies are independent of ignorance, increasing the publicity too much can have negative effects on robustness:

\begin{prop}
    If $\gamma\frac{\beta_1}{\beta_0}>1$, increasing publicity decreases robustness of the most robust full-participation equilibrium.
\end{prop}
Loosely, this says that increasing the publicity of actions faster than people adapt their behaviour under consideration of this publicity necessarily weakens norms.
\begin{proof}
    The discounted reproduction rate of the ignorant agents is initially $\gamma\overline{\delta}_1$. Due to the publicity shock, it increases to $\frac{\beta_1}{\beta_0}\gamma\overline{\delta}_1>\overline{\delta}_1$. Thus other agents have a strict preference to contribute (this is not an equilibrium of the full information game). However, the reproduction rate also increases by a factor of $\frac{\beta_1}{\beta_0}\gamma>1$ and thus results in a less robust (or possibly non-robust) equilibrium.
\end{proof}

\section{Analysis of General Strategies}\label{sec:strat.gen}
In this section, we extend our analysis to general strategy spaces, and show that our analysis of Finite Deviation strategies is with little loss in generality. Formally, we present two arguments: (i) a natural class of `monotone' population strategies are essentially equivalent to Finite Deviation strategies (Proposition \ref{prop:monotone}) and (ii) Finite Deviation PSE form the existence/robustness frontier (Proposition \ref{prop:mono.suff}).

\subsection{Impulse Response and Strategy Equivalence}\label{app:imp.eq}
For general strategies the key question for an agent is how their choice of action is likely to influence the future actions of others. This is summarized by the \textbf{impulse response}, defined as a sequence $\nu(\mu)\in[-1,1]^\mathbb{N}$ where $\nu_n$ is the additional probability an agent observing an action deviates $n$ period later if it is a contribution rather than a deviation. Thus for the pure $n$-Deviation strategy $\sigma_n$, we have $\nu_m(\sigma_n)=1\{m\le n\}$. 

Note that a full-participation PSE must have $\nu(\mu)\ge 0$ --- agents always contribute if they observe a contribution, so observing a deviation cannot increase the probability of contributing.
\begin{defn}
    A population strategy $\mu$ is \textbf{positive} if changing any observation from $D$ to $C$ weakly increases the probability of contributing. A strategy is \textbf{monotone} if $n\mapsto \nu_n(\mu)$ is decreasing.
    
    A strategy \textit{space} $\langle\Sigma,\sigma^+,\hat{a}\rangle$ is \textbf{monotone} if all contributing states continue contributing if they do not observe a deviation.
\end{defn}
These are natural properties: a positive strategy means that observing a contribution has a weakly positive influence on others' actions\footnote{This is a population level property: we allow for some agents to react opposite to an observation, as long as they are outweighed by others who are influenced positively by the observation}, while a monotone strategy actions' influence decrease over time. 


The discounted reproduction rate and reproduction rate generalize to
\begin{align*}
D(\mu):=&\beta\sum_{n=1}^\infty \delta^n\nu_n(\mu),&
\Re(\mu):=&\beta\sum_{n=1}^\infty \nu_n(\mu),
\end{align*}
which can then be used to find equilibria, evaluate stability, and compare robustness (under assumptions in Appendix \ref{app:proof}).

Two population strategies $\mu,\mu'$ are \textbf{impulse equivalent} if they have the same impulse response $\nu(\mu)=\nu(\mu')$. In this case we have $D(\mu)=D(\mu')$ and $R(\mu)=R(\mu')$, making such stationary strategies equivalent for our analysis.

\begin{prop}\label{prop:monotone}
The following are equivalent for a stationary strategy $\mu$:
\begin{enumerate}[label = (\arabic*)]
    \item $\mu$ is monotone,
    \item $\mu$ is impulse equivalent to a strategy $\mu'$ in a monotone strategy space,
    \item $\mu$ is impulse equivalent to a distribution over $\{0,1,\dots,\infty\}=:\overline{\mathbb{N}}_0$ (ie. Finite Deviation strategies and Grim Trigger).
\end{enumerate}
\end{prop}
Thus our analysis of Finite Deviation strategies is comprehensive of all stationary strategies satisfying this natural monotonicity property.\footnote{Note that $\mu$ and $\mu'$ need not have the same deviation rate to be impulse equivalent. Adding this constraint to equivalence does not affect the result for full-participation equilibria, but may affect the result for partial-participation equilibria $\mu$. In particular, if $\lambda$ and $\beta$ are high enough that agents are likely to observe multiple deviations, then $\nu(\mu)\approx 0$ for Finite Deviation strategies. A strategy that deviates if at least (say) half of observations are deviations can increase the odds than a particular observation is pivotal, thereby obtaining an impulse response unobtainable by Finite Deviation strategies.
}


Another question is whether any of our existence results are threatened by the consideration of exotic non-monotone strategies. The answer is no:
\begin{prop}\label{prop:mono.suff}
\begin{enumerate}[label = (\alph*)]
    \item If there exists a PSE $\mu\neq \sigma_{AD}$, there exists a monotone full-participation PSE.
    \item If there exists a robust full-participation PSE $\mu$ that is not monotone, there exists a Finite Deviation PSE that is strictly more robust.
    \item If there exists a partial-participation PSE $\mu$ that is positive and locally stable, there exists a Finite Deviation PSE that is globally stable and robust.
\end{enumerate}
\end{prop}
The first two are intuitive results: non-monotone strategies involve delaying deviations, which diminish the discounted reproduction rate without lowering the reproduction rate. 

The last statement says that full-participation PSE are in theory easier to obtain than positive partial-participation PSE. Thus even in settings where near full-participation is implausible --- perhaps there is a large population of selfish agents --- full-participation equilibria remain a useful tool for understanding steady state behaviour and incentives. With regards to (c), we could also relax stationarity to think of partial-participation as a transient phenomenon as norms oscillate between high and low participation. By foregoing stability, transient norms may be easier to support than robust full-participation norms, but more difficult to support than the fragile Grim Trigger equilibrium. The analysis of such equilibria is beyond the scope of this paper.

\section{Discussion}\label{sec:conc}
The contribution of this paper is to describe a novel mechanism supporting social norms through a combination of social preferences and social influence. This highlights some novel forces underlying these norms. 

We have seen that altruism and patience are substitutes when it comes to supporting robust norms --- only a little of one is needed, provided the other can compensate. Moreover, the publicity of actions is irrelevant beyond the modest level of $\beta=1$. An important question we do not touch on is the selection of equilibria, this question requires modelling a mechanism that drives heterogeneous influence. A work-in-progress develops one such model --- discussed further in Section \ref{sec:mech.ag}.

We dedicate the remainder of this section to discussing interpretations of the model's features.

\subsection{Continuum Population}\label{sec:cont.pop}

A central innovation in this model is that agents care about others' \textit{individual} actions (rather than an average) within a continuum population. The continuum population assumption plays three roles in our analysis:

First, it ensures that influence is unbounded, a single deviation can eventually influence arbitrarily many agents to deviate, leading to huge externalities. This is only relevant for fragile equilibria (e.g. Proposition \ref{prop:exist}). For robust equilibria, an action's influence is capped to (on average) one individual in each subsequent period, so continuum populations are not necessary to generate these externalities. 

The second, more important, role is to ensure the property of Observation Independence. In any finite population model with influence, observations will not be independent. However, for a robust equilibrium, where an action influences (on average) less than one individual in each subsequent period, the correlation between observations is inversely proportional to population size. Thus, even if people observe/interact with as much of 1\% of the population at a time, it may natural think of these interactions as independent. From such a strategic perspective, Observation Independence can be thought of as a property of bounded rationality, where agents do not update their beliefs about others' strategies based on their observations.

Lastly, operating with a continuum of agents allows the application of the law of large numbers with abandon when analyzing robustness and stability.

\subsection{Altruism \& Extensions}\label{sec:altruism.FAQ}
The essential preference property for our model is that agents care about their influence on the actions of other individuals through the factor $\alpha$. We motivate this through altruism, but this could also occur with material preferences if actions impose a global\footnote{This can be weakened to diffuse externalities, see the model of pollution in the Online Appendix.} externality to others (e.g. pollution). Similar $\alpha$ could also arise through other moral concerns not measured by welfare altruism.


We describe the preference in eq. \ref{eq:pref} in the language of \textit{perfect altruism} (\cite{A90}) where agents do not care about where contributions come from, but only the gain from contributions. This contrasts with \textit{warm glow altruism} where agents irrationally care more about the contributions they make. 

It is not clear to what extent agents feel a warm glow from contributions one influences (or expects to influence), however it is reasonable to expect it may be less than the warm glow one feels from one's own contribution. By simply reinterpreting the effective altruism $\alpha$, our model is compatible with such altruism:

Suppose a contribution costs $c$ and produces welfare $c+g$ ($g$ being the welfare gain). Agents value the welfare produced by their own contribution at $a_s>0$, and the welfare gained through others' contributions at $a_o\le a_s$ (equality being perfect altruism). They then prefer the path $p^C$ to $p^D$ if
\begin{equation*}\begin{split}
    &\sum_{\tau\ge t} \delta^{\tau-t}\Big(c(p^C_{i,\tau}-p^D_{i,\tau})+a_s(c+g)(p^D_{i,\tau}-p^C_{i,\tau})+\sum_{j\neq i}a_og(p^D_{j,\tau}-p^C_{j,\tau})\Big)\\
    =&\sum_{\tau\ge t} \delta^{\tau-t}\Big(\big(c-a_s(c+g)+a_og)(p^C_{i,\tau}-p^D_{i,\tau})+\sum_{j\in\mathbb{I}}a_og(p^D_{j,\tau}-p^C_{j,\tau})\Big)
\end{split}
\end{equation*}
By normalization, this is equivalent to eq. \ref{eq:pref} with $\alpha:=\frac{a_o g}{c(1-a_s)-(a_s-a_o)g}$. Our analysis assumes $0<\alpha<1$, which translates to (i) $a_s(c+g)<c$, meaning agents are not willing to contribute without influence (otherwise contributing can be sustained myopically), and (ii) they are not purely egoist, ie. $a_og> 0$, they care about their influence/the actions of others. As mentioned, the second summation index can be restricted to actions that the agent influences, in which case the $a_o$ parameter can be interpreted as a `second-order' warm glow, that agents feel from acts they influence. 

This model is also compatible with there being some risk of punishment from deviating from the norm, this is equivalent to decreasing the cost $c$ of contributing, thereby increasing $\alpha$ and enabling more robust norms. We merely require that this punishment is sufficiently mild to preserve $\alpha<1$, so that deviating remains myopically optimal.

In the Online Appendix, we formulate the general properties that we assume about the preference necessary to obtain our results. These properties allow agents to care less about the actions of more (socially) distant agents --- as described by a model of diffuse pollution. 

We model preferences as linear. This `marginal' approach to altruism is justified in robust equilibria as agents' influence is bounded and the state of the world does not evolve --- our results hold as long as there is some bounded influence that an agent could have that would make them willing to contribute. In reality, we might expect agents' willingness to contribute to depend on their own experience --- if they are consistently the victim of others' selfishness they may become more selfish. Such a preference is consistent with our equilibria as it rationalizes social influence, but violates our principle of mechanism agnosticism described in Section \ref{sec:mech.ag}.

Other extensions have agents interact strategically (e.g. through prisoner's dilemmas), or otherwise face different incentives based on the aggregate behaviour of agents (e.g. pressure to conform) --- these considerations correspond to the effective altruism $\alpha$ depending on the deviation rate $\lambda$. If altruism is increasing in $\lambda$, then agents are more willing to contribute the more others deviate, making it easier to support partial norms. Consequently, Proposition \ref{prop:mono.suff}c) may no longer hold, and the robustness of equilibria with $R(\mu)=1$ (studied in Appendix \ref{app:proof}) can be affected, but otherwise our results carry over. 

In this paper we model norms as a single action; there are many settings where it may be helpful to understand norms as involving multiple actions, but beyond the scope of this paper. To give a few examples: there are asymmetric norms where expectations depends on one's fixed role (e.g. tipping viewed as part of server--customer norms); there are also thematic norms (e.g. norms of fairness) composed of multiple `micro-norms', where observing a violation of expected behaviour in one setting may influence negative behaviour in related but distinct settings; lastly, some norms may be supported under the threat of strangers choosing a costly `punishment' action (what \citet{FG01} refer to as \textit{altruistic punishment}), this punishment can be motivated by its capacity to influence future behaviour in the deviating agent and their future observers through a similar ripple effect. This last case may be of particular interest, as a means of leveraging the influence of altruistic agents to incentivize moral behaviour even in selfish agents.

\subsection{Relations with Morality}

To violate a norm in our model is to break the trust of others, influencing observers to further violate the norm. Within this social context, to contribute is not just to make a `contribution' \textit{per se}, but is also to bolster the norm through one's influence; likewise to violate the norm is not just to obtain personal benefit, but it is also to weaken the norm\footnote{Strengthening and weakening should be understood in a local sense, consistent with robustness.}. In this way, the social trust and expectations of others imposes a moral duty to conform to the norm.

Despite operating on consequentialist reasoning, our model mirrors deontological ethics; as $\delta\rightarrow 1$, it often produces outcomes consistent with Kant's categorical imperative: ``Act only in accordance with that maxim through which you can at the same time will that it become a universal law'' (\cite{Kant}).
 
To see this more generally, suppose agents interact through a symmetric game with binary actions $\{C,D\}$, and denote the effective altruism when the deviation rate is $\lambda$ by $\alpha_\lambda$ (determined by a level of pure altruism $a>0$). With enough patience, Ripple Effects can support a norm of contributing if $\alpha_0>0$. Supposing the cost of contributing is positive when everyone else contributes (otherwise contributing is a myopic equilibrium), then the sign of $\alpha_0$ is determined by the welfare gain from contributing when everyone else contributes, denoted $g_0$. 

Restricting our attention to symmetric games, define $u_K(\lambda)$ to be the expected utility of an agent under the symmetric mixed strategy $(1-\lambda)C\oplus \lambda D$. If $\lambda=0$ maximizes $u_K$, we say $(C,C,\cdots)$ is a \textit{Kantian equilibrium} (\cite{Roemer19}). Moreover, the welfare gain from contributing satisfies $g_\lambda=-u_K'(\lambda)$, and our condition for Ripple Effects to support $C$ becomes $u_K'(0)<0$. Thus the condition for social influence to support a norm is related to the first order condition for a Kantian equilibrium. 

There are two subtle distinctions between Ripple Effect norms and Kantian equilibria. First, $(C,C,\cdots)$ may be a Kantian equilibrium with $u_K'(0)=0$, as in Game \ref{tab:kant}. This knife-edge case is not substantial: since $\alpha_\lambda$ is increasing in $\lambda$, there exist partial participation PSE with almost full participation as $\delta\rightarrow 1$. A more significant difference is how Ripple Effects can support outcomes that are not Kantian equilibria if $\lambda=0$ is only a \textit{local} maximum of $u_K$, as in Game \ref{tab:nokant}. In this case $(C,C,\cdots)$ might be called `marginally-Kantian', as it satisfies a marginal version of Kant's categorical imperative: one should act as one wills the marginal fraction of the population to act.\footnote{This condition can be further relaxed in multi-action settings, as in the subsequent paragraph. In this case it suffices that $C$ is marginally Kantian when restricting to a single `selfish' alternative $D$ that has a lower material cost to the agent than $C$ when everyone else chooses $C$. When agents are patient, we can then construct an equilibrium where agents are only ever influenced to choose $C$ or $D$ and actions that are less costly than $D$ carry more extreme negative influence.}
\begin{table}
\centering
\begin{subtable}{.4\textwidth}
\centering
\begin{tabular}{r|c|c}
     & C&D \\\hline
     C&2,2&0,4\\\hline
     D&4,0&1,1
\end{tabular}
\caption{ $\alpha_\lambda\propto \frac{2\lambda}{2-\lambda}$, $u_K(\lambda)=2-\lambda^2$}
\label{tab:kant}
\end{subtable}\hspace{2cm}
\begin{subtable}{.4\textwidth}
\centering
\begin{tabular}{r|c|c}
     & C&D \\\hline
     C&2,2&0,3\\\hline
     D&3,0&4,4
\end{tabular}
\caption{ $\alpha_\lambda\propto\frac{1-6\lambda}{1+3\lambda}$, $u_K(\lambda)=2-\lambda+3\lambda^2$}
\label{tab:nokant}
\end{subtable}
\caption{Examples of $C$ as (a) the unique Kantian equilibrium, but not a norm supportable by Ripple Effects; (b) a norm supportable by Ripple Effects, but not a Kantian equilibrium.}
\end{table}

Game \ref{tab:nokant} demonstrates how Ripple Effects can support `sub-optimal' moral outcomes. This may also occur if agents have three actions available: a selfish action $D$, a moral action $C$ producing welfare $g_C$, and an `extra moral' action $E$ producing welfare $g_E>g_C$. If $E$ violates trust and significantly influences others to choose $D$, we can support $C$ as a norm when agents are sufficiently patient. Such an equilibrium may be more natural than appears, and relates to certain controversies in consequentialist ethics:

\citet{singer72} proposes a well-known thought experiment where one walks past a child drowning in muddy water, and observes that there is a moral obligation to save the child (action $C$) rather than ignoring the child ($D$), even if saving the child will ruin one's clothes. Singer then argues that one has the same obligation to donate money to (say) famine relief that saves children rather than spend that money on luxury items. \citet{appiah06} points out that if one is wearing (sufficiently) expensive clothing, then the same moral calculus concludes, repugnantly, that one has the same obligation to \textit{not} save the child, but rather sell one's unmuddied clothes so that the proceeds can be donated to effective famine relief (representing $E$), with the end of saving multiple children.

While the direct consequences of $E$ may be judged more moral than $C$ within this calculus (and we advance no argument to the contrary), our mechanism suggests that if $E$ involves a betrayal of social trust, it may have a significant negative influence on the actions of others, potentially negating the direct moral consequences of $E$. (This influence need not be limited agents' behaviour in identical situations, but may extend to other moral contexts.)

\subsection{Mechanism Agnosticism}\label{sec:mech.ag}
The model is constructed in such that agents' incentives do not change with their observations. This should be interpreted as agnosticism about \textit{why} and \textit{how} agents are influenced by observed behaviour, allowing them to be influenced through indifference instead.

While this prevents us from discussing equilibrium selection, it also allows a flexibility that is a strength of the model. While it is an established fact that people are influenced by observed behaviour\footnote{See the literatures on social contagion, conformity, and conditional cooperation, cited in the literature review.  

Beyond behavioural motivations such as conformity, real world social influence can plausibly be rationalized by: social learning of $\alpha$; social learning of others' strategies (relaxing observation independence); or structural incentives (e.g. lone deviators may be likely to face a mild punishment --- such as mild scorn or social disapproval --- while the presence of a `nearby' deviator may mask one's own deviation or alleviate the punishment).} \textit{how} and \textit{how much} people are influenced varies, both across settings and heterogeneously within settings --- it is not always clear how individuals resolve the indifferences of our equilibria. This paper is thus content to describe the range of possible (full-participation) equilibria, and make the conclusion that strategic (ie. influence) homogeneity is good for robustness. 

To select an equilibrium in this framework is to make assumptions about how and how heterogeneously influence occurs. For example, Proposition \ref{prop:monotone} details how our focus on Finite Deviation equilibria (and their equivalence class) is formally associated with the assumption that influence weakens over time (ie. the causal effect of an observation $n$ periods ago on one's action today is decreasing in $n$).

Adding any sort of dynamism to the model (e.g. relaxing observation independence, or adding history dependence to preferences) would strictly rationalize influence, and thus select equilibria. But, unless this mechanism operates heterogeneously, it will, perhaps unrealistically, select an equilibrium with homogeneous strategies (e.g. the concentrated equilibrium, up to impulse equivalence). 

As an example, suppose agents' effective altruism $\alpha$ decreases after observing a deviation (describing a behavioural source of influence) and increases to some steady state level after observing contributions. Then (for high $\delta$) the unique (non-myopic) PSE will solve eq. \ref{eq:DDR} for the level of effective altruism that occurs after observing a single deviation. Agents will thus have a strict preference for contributing if they did not observed a deviation the last period and be willing to randomize between deviating and contributing after observing a deviation, reproducing the concentrated 1-Deviation equilibrium. A work-in-progress augments this model with more general heterogeneous and history-dependent preferences (attributable to costs $c$ or altruism $a$), allowing heterogeneous strategies to be attributed to heterogeneous preferences. 

\printbibliography

\appendix
\section{Observation Independence}\label{app:obs.ind}
The key property provided by assuming agents form a continuum is so that an agents' observations are statistically independent of their previous actions and observations. This is an intuitive property in large population settings: every action is a small part of society and we might expect it to thus only have a small effect on how the future occurs. Using continuum population introduces this as a property of our model.

Let $O(i,t)$ be the set of agents whose action agent $i$ observes in period $t$. Their action next period $a_{i,t+1}$ could potentially be influenced (directly or indirectly) by any agent in the following set (illustrated in Figure \ref{fig:cone}).
\begin{defn}
    The \textbf{past cone of influence} $\mathcal{C}_{i,t}^-\subseteq \mathbb{I}\times\mathbb{T}$ of an action with index $(i,t)$ is the smallest set of indices that
    \begin{enumerate}
        \item contains their previous observations: if $j\in O(i,\tau)$ and $\tau<t$ then $(j,\tau)\in\mathcal{C}_{i,t}^-$.
        \item is transitive: if $(i',t')\in\mathcal{C}_{i,t}^-$ then $\mathcal{C}_{i',t'}^-\subseteq \mathcal{C}_{i,t}^-$.
    \end{enumerate}
    The \textbf{future cone of influence } $\mathcal{C}_{i,t}^+\subseteq \mathbb{I}\times\mathbb{T}$ of a stage player $(i,t)$ is the set of indices $(i',t')$ such that $(i,t)\in\mathcal{C}_{i',t'}^-$.
\end{defn}
    
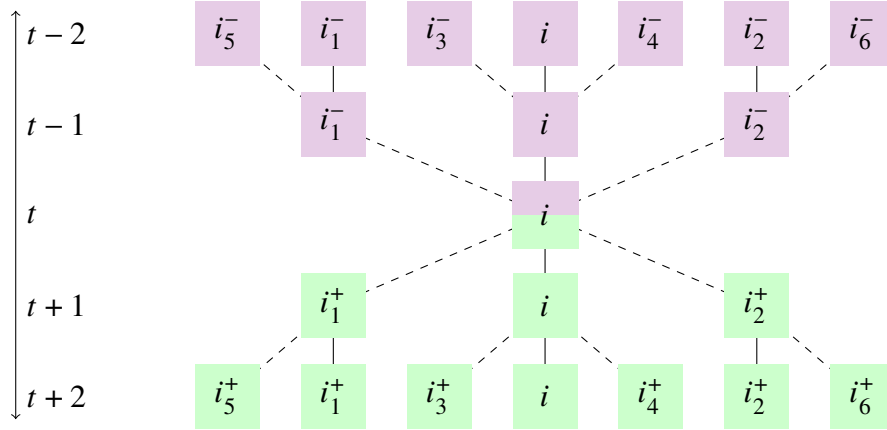
\begin{figure}
    \centering
\begin{tikzpicture}[
    Cone/.style={rectangle, fill = green!20, text = black,minimum size=8.5mm},
    Conep/.style={rectangle, fill = violet!20, text = black,minimum size=8.5mm},
    state/.style={ 
        rectangle split,
        rectangle split parts=2, 
        rectangle split part fill={violet!20,green!20}
        },
    state+/.style={
        state,
        rectangle split every empty part={},
        rectangle split empty part width=6mm,
        rectangle split empty part height=1.5mm}]
        \draw[<->] (-7,-2.7) -- (-7,2.7);
        \node[anchor = west] at (-7,1.2) {$t-1$};
        \node[anchor = west] at (-7,2.4) {$t-2$};
        \node[anchor = west] at (-7,0) {$t$};
        \node[anchor = west] at (-7,-1.2) {$t+1$};
        \node[anchor = west] at (-7,-2.4) {$t+2$};
        \draw[] (0,2.8) -- (0,0);
        \draw[ dashed] (2.8,1.2)--($(0,0)+(0,0)$);
        \draw[ dashed] (-2.8,1.2)--($(0,0)+(-0,0)$);
        \draw[ dashed] (1.4,2.4)--($(0,1.2)+(0,0)$);
        \draw[ dashed] (-1.4,2.4)--($(0,1.2)+(-0,0)$);
        \draw[] (-2.8,2.4) -- (-2.8,1.2);
        \draw[] (2.8,2.4) -- (2.8,1.2);
        \draw[ dashed] (-4.2,2.4)-- ($(-2.8,1.2)+(-0,0)$);
        \draw[ dashed] (4.2,2.4)--($(2.8,1.2)+(0,0)$);
        \node[Conep] at (1.4,2.4) {$i^-_4$};
        \node[Conep] at (-1.4,2.4) {$i^-_3$};
        \node[Conep] at (2.8,1.2) {$i^-_2$};
        \node[Conep] at (-2.8,1.2) {$i^-_1$};
        \node[Conep] at (-2.8,1.2) {$i^-_1$};
        \node[Conep] at (-2.8,2.4) {$i^-_1$};
        \node[Conep] at (2.8,2.4) {$i^-_2$};
        \node[Conep] at (-4.2,2.4) {$i^-_5$};
        \node[Conep] at (4.2,2.4) {$i^-_6$};
        \node[Conep] at (0,2.4) {$i$};
        \node[Conep] at (0,1.2) {$i$};
        \draw[] (0,-2.4) -- (0,0);
        \draw[ dashed] (0,0)--($(2.8,-1.2)+(-0,0)$);
        \draw[ dashed] (0,0)--($(-2.8,-1.2)+(0,0)$);
        \draw[ dashed] (0,-1.2)--($(1.4,-2.4)+(-0,0)$);
        \draw[ dashed] (0,-1.2)--($(-1.4,-2.4)+(0,0)$);
        \node[Cone] at (1.4,-2.4) {$i^+_4$};
        \node[Cone] at (-1.4,-2.4) {$i^+_3$};
        \node[Cone] at (2.8,-1.2) {$i^+_2$};
        \node[Cone] at (-2.8,-1.2) {$i^+_1$};
        \draw[] (-2.8,-2.4) -- (-2.8,-1.2);
        \draw[] (2.8,-2.4) -- (2.8,-1.2);
        \draw[ dashed] (-2.8,-1.2)--($(-4.2,-2.4)+(0,0)$);
        \draw[ dashed] (2.8,-1.2)--($(4.2,-2.4)+(-0,0)$);
        \node[Cone] at (2.8,-1.2) {$i^+_2$};
        \node[Cone] at (-2.8,-1.2) {$i^+_1$};
        \node[Cone] at (-2.8,-2.4) {$i^+_1$};
        \node[Cone] at (2.8,-2.4) {$i^+_2$};
        \node[Cone] at (-4.2,-2.4) {$i^+_5$};
        \node[Cone] at (4.2,-2.4) {$i^+_6$};
        \node[state+={null}] at (0,0) {};
        \node at (0,0) {$i$};
        \node[Cone] at (0,-2.4) {$i$};
        \node[Cone] at (0,-1.2) {$i$};
        
    \end{tikzpicture}
    \caption{The past/future cone of influence of stage player $(i,t)$ in purple/green extending two periods. Time increases as we move down the page, solid lines indicate continuations of the same individual, dashed lines indicate the player below observed the action of the player above. Observation Independece implies that $i,i^-_1,\dots,i^-_6,i^+_1,\dots,i^+_6$ are distinct individuals.}
    \label{fig:cone}
\end{figure}
    Note that every action outside the future cone of influence $\mathcal{C}_{i,t}^+$ will be independent of the stage player $(i,t)$'s action.

We assume that agents do not believe that they have any common influence with those they observe.
\begin{axiom}[Observation-Independence]
    If $(i_N,t+N)\in\mathcal{C}_{i_0,t}^+$ then w.p. $1$, there exists a unique sequence $(i_0,i_1,\dots, i_N)$ such that each agent observes the preceding action: $i_{n}\in O(i_{n+1},t+n+1)$ for $n=0,\dots,N-1$.
\end{axiom}
This says that there is only one path through which an agent can influence any other. Thus if an agent influences two agents, there is no way that those two agents can both influence a third agent. 

Observation independence implies that the past actions of an agent are statistically independent of their future observations, so agents never \textit{observe} any consequence of their actions; as well as that the cone expands exponentially over time.\footnote{A trivial extension of the model allows some observations of agents to not be independent as long as the agent can discriminate between these observations and thus not react/be influenced by them. For example, agents may observe how they are treated by strangers, as well as how various acquaintances treat others, but their actions are only influenced by the former.}

\section{Supplemental Proofs}\label{app:proof}
\subsection*{Stability and Robustness}
This section proves Lemma \ref{lem:rob}. We break the proof into three parts, examining conditions for stability in Lemma \ref{lem:stab}, conditions for robustness (and ordering of robustness) in Lemma \ref{lem:rob2}, before studying in specifically the edge case $R(\mu)=1$ in Proposition \ref{prop:rob.edge} (which also is necessary to conclude that robustness implies stability in the Finite Deviation scenario).

We first extend our definition of impusle response to partial participation settings. In this case the impulse response of a strategy depends on the deviation rate $\lambda_\mu$. We represent the population strategy as a mixed strategy $q^\mu:(\{C,D\}^\beta)^{<\infty})\rightarrow[0,1]:h\mapsto q_h^\mu$ mapping sequences of observations to a probability of contributing. The deviation rate $\lambda_\mu$ then defines a probability distribution over histories of observations, and in particular over length $n$ histories $h_{-1}^n$ that omit the first observation from the initial period\footnote{For simplicity we assume observations are ordered although agent strategies cannot differentiate between them.}. Let $Ch_{-1}^n$ (resp. $Dh_{-1}^n$) be the history where the omitted observation is a contribution (resp. deviation). The impulse response is then defined
$$
\nu_n(\mu):=\E[h_{-1}^n]{q^\mu(Ch_{-1}^n)-q^\mu(Dh_{-1}^n)}.
$$

A key tool in our analysis is the \textbf{induced deviation rate} from a sequence of deviation rates associated with a strategy. Formally, let $\lambda_{-t,k}$ be a sequence of probabilities determining the probability of each observation being a deviation in the previous $N$ periods. This induces a probability over these observations $h$, and the induced deviation rate given this probability distribution is then
$$
\Lambda_\mu^N(\lambda_{-t,k}):=\E[h]{1-q^\mu(h)}.
$$
Letting $F_\beta$ be the distribution that the number of observations is drawn from (if random) and $\overline{b}:=\max\supp{F_\beta}$, the induced deviation rate will be a degree $\overline{b}N$ polynomial in $\lambda_{-t,k}$ with coefficients bounded by a function of $\overline{b}N$. For a stationary strategy $\mu$ with deviation rate $\lambda_\ast$, we have $\Lambda_\mu^N(\lambda_\ast)=\lambda_\ast$ for any $N$ (where we interpret a scalar $\lambda$ as the vector $\lambda_{-t,k}\equiv\lambda$ in the argument of $\Lambda_\mu^N$). Moreover, 
$$
\pderiv{\Lambda_\mu^N}{\lambda_{-n,1}}(\lambda_\ast)=\nu_n,
$$
for any $N\ge n$. From this property, $(\Lambda_\mu^N)'(\lambda_\ast)\rightarrow R(\mu)$ as $N\rightarrow \infty$. For positive strategies this convergence is from below and $\Lambda_{\mu}^N$ is weakly increasing in the deviation rate for each observation, consequently if $\underline{\lambda}\le\lambda_{-t}\le \overline{\lambda}$ for all $t$ then $\Lambda_\mu^N(\underline{\lambda})\le\Lambda_\mu^N(\lambda_{-t})\le \Lambda_\mu^N(\overline\lambda)$, where $\lambda_{-t}$ represents $\lambda_{-t,k}\equiv\lambda_{-t}$ for any $k,t$.

As a concrete example, for the pure $n$-Deviation strategy, and constant path $\lambda_{-t}\equiv \lambda$,
\begin{equation}\label{eq:FD.ind.dev}
\Lambda_{n}^N(\lambda)=\idotsint \left[1-(1-\lambda)^{b_1+\cdots+b_{n\wedge N}}\right]\,dF_\beta(b_1)\cdots dF_\beta(b_{n\wedge N}),
\end{equation}
where $n\wedge N:=\min\{n,N\}$. This is a concave function of $\lambda$, strictly so whenever $n\wedge N\ge 2$ or $\overline{b}\ge 2$ (with second derivative decreasing in $N$). The induced deviation for a population strategy $\mu$ can be obtained by taking an expectation over $n\sim\mu$.

If $n_k$ is the deviating state associated with the $n$-Deviation strategy that moves to the contributing state only after observing $k$ periods of contributions, we have $\Lambda_{n_k}^N(\lambda)=\Lambda_n^N(\lambda)\boldsymbol{1}\{N\ge k\}$.

\begin{customlem}{1A}\label{lem:stab}
    A PSE $\mu$ with deviation rate $\lambda_\ast$ is not locally stable if $\sup_N\frac{\Lambda_\mu^N(\lambda)-\lambda_\ast}{\lambda-\lambda_\ast}>1$ in a neighbourhood of $\lambda_\ast$.

    A PSE $\mu$ supported on finite memory positive strategies is locally stable if $\sup_N\frac{\Lambda_\mu^N(\lambda)-\lambda_\ast}{\lambda-\lambda_\ast}<1$ in a neighbourhood of $\lambda_\ast$. Such a PSE is globally stable if this inequality holds for any $\lambda\neq\lambda_\ast$.

    A Finite Deviation PSE is globally stable iff it is locally stable.
\end{customlem}
\begin{proof}
\textbf{Necessity:} Suppose $\frac{\Lambda_\mu^N(\lambda)-\lambda_\ast}{\lambda-\lambda_\ast}> 1$ over an interval $]\lambda_\ast,\lambda_\ast+\epsilon']$ for some $N,\epsilon'>0$. In this case, for any constant sequence $\lambda_{-t}\equiv \lambda_\ast+\epsilon$ of length $N$, with $\epsilon\le\epsilon'$, we have
$$
\Lambda_\mu^N(\lambda_{-t})> \lambda_\ast+\epsilon.
$$
Thus the deviation rate will leave the $\epsilon$ ball around $\lambda_\ast$. If $\lambda_\ast=1$ we can repeat the same argument with with negative shocks to the deviation rate.

\textbf{Sufficiency:} Finite memory strategies only depend on the previous $N$ periods, and as long as the deviation rate converges to $\lambda_\ast$, the strategies will converge as well. We will show this for any population strategy supported on finite memory positive strategies.\footnote{By taking a monotone envelope of $\Lambda_{\le N}^N$ the proof can be modified to cover the case, for infinitely many $N$, the population strategy of agents with memory less than $N$ and the population strategy of agents with memory larger than $N$ are both positive.}

Suppose $0\le\sup_N\frac{\Lambda_\mu^N(\lambda)-\lambda_\ast}{\lambda-\lambda_\ast}<1$ over the interval $]\lambda_\ast,\lambda_\ast+\epsilon]$, then for any sequence $\lambda_{-t}$ with $\overline{\lambda}_0:=\max_{-t}\{\lambda_{-t}\}\le \lambda_\ast+\epsilon$, we have
$$
\Lambda_\mu^N(\lambda_{-t})\le \Lambda_\mu^N(\overline\lambda_0)<\overline{\lambda}_0.
$$
Thus the deviation rate will not exceed the original neighbourhood of $\lambda_\ast$. 

We decompose $\Lambda^N_\mu=\Lambda^N_{\le N}+\Lambda^N_{>N}$ where $\Lambda^N_{\le N}$ is the induced deviation rate attributed to agents with memory less than $N$ and $\Lambda^N_{> N}$ is the induced deviation rate of the $\epsilon_N$ measure of agents with memory larger than $N$. As $N\rightarrow\infty$ we have $\epsilon_N\downarrow 0$ and $\Lambda^N_\mu-\Lambda^N_{\le N}\downarrow0$. 
Note that if $\overline{\lambda}_{m,N}$ is an upper bound on the deviation rate of $\lambda_t$ between period $(m-1)N$ and $mN$, then
$$
\Lambda^{mN}_\mu(\lambda_{t})\le \Lambda_{\le N}^{mN}(\lambda_{t})+\epsilon_N\le \Lambda^N_{\le N}(\overline{\lambda}_{m,N})+\epsilon_N
$$
Let $\overline{\lambda}_0:=\max_{-t}\{\lambda_{-t}\}<\lambda_\ast+\epsilon$ and suppose $N$ is such that
$$
\frac{\Lambda^N_{\le N}(\overline\lambda_0)-\lambda_\ast}{\overline{\lambda}_0-\lambda_\ast}<\frac{\overline{\lambda}_0-\lambda_\ast-\epsilon_N}{\overline{\lambda}_0-\lambda_\ast}
$$
Since the left side has supremum (over $N$) below one and the right side has limit 1, there is a sequence of such $N$ converging to infinity.

We then recursively construct a sequence of upper bounds
$$
\overline{\lambda}_{m+1,N}:=\Lambda^N_{\le N}(\overline{\lambda}_{m,N})+\epsilon_N
$$
decreasing in $m$. The iterating function $ \Lambda^N_{\le N}(\cdot)+\epsilon_N$ is a contraction, mapping $[\lambda_\ast,\overline{\lambda}_0]$ to $[\lambda_\ast,\overline{\lambda}_0]$, this sequence of upper bounds converges to a fixed point $\lambda_\ast^N\ge \lambda_\ast$. As $N\rightarrow\infty$, the fixed point of $ \Lambda^N_{\le N}(\cdot)+\epsilon_N$ converges to $\lambda_\ast$. 

By repeating the exercise for a lower bound (using $\Lambda^N_{>N}=0$ as a lower bound), we obtain a sequence of lower bounds converging to $\lambda_\ast$.

\textbf{Equivalence of Local and Global Stability:} Local stability of a Finite Deviation PSE requires $(\Lambda_\mu^N)'(0)\le 1$ for all $N$. From the proofs above it is sufficient to show this implies $\sup_N\Lambda_\mu^N(\lambda)/\lambda<1$. 

Recalling the concavity of $\Lambda_\mu^N$, we are done if for any locally stable $\mu$, for all $N$ either (a) $(\Lambda_\mu^N)'(0)<1$ or (b) $(\Lambda_\mu^N)'(0)\le 1$ (implied by local stability) and $(\Lambda_\mu^N)''(\lambda)<0$ for all $\lambda<1$ (for uniformity, note that $(\Lambda_\mu^N)''(\lambda)$ is weakly decreasing in $N$ for Finite Deviation strategies). 

If $R(\mu)<1$, then (a) holds for all $N$ and we are done. So suppose $R(\mu)=1$.

If $\overline{b}>2$ then $(\Lambda_\mu^1)''(\lambda)<0$ and (b) holds for all $N$, and we are done. So suppose $\overline{b}=1$.

Now suppose (a) fails for $N=1$, or (b) fails for $N\ge 2$. In the former case, we have $\beta\nu_1(\mu)=(\Lambda_\mu^1)'(0)=1=R(\mu)$ implying $\mu_1=\beta=1$. In the latter case, we have $(\Lambda_\mu^2)''(\lambda)=0$ for some $\lambda<1$ which only occurs if $\max\supp{\mu}\le 1$, implying $R(\mu)=\beta\mu_1=1$ and again $\mu_1=\beta=1$. But if $\mu=\sigma_1$ and $b=1$ w.p. 1, then all agents are copying their previous observation: $\Lambda_1^N(\lambda)=\lambda$, and any fluctuation in the initial deviation rate persists indefinitely, precluding local stability.
\end{proof}

\begin{customlem}{1B}\label{lem:rob2}
    A PSE $\mu$ is not behaviourally robust if $R(\mu)>1$.

    If a PSE $\mu$ is supported on finite memory strategies, it is robust if $R(\mu)<1$.

    A robust Finite Deviation PSE $\mu$ is strictly more behaviourally robust than the robust PSE $\mu'$ if
    $$
    R(\mu)<R(\mu').
    $$    
\end{customlem}
\begin{proof}
In the following proof we adopt the shorthand $\mu_{\text{AD}}:=\mu[\sigma_{\text{AD}}]$.

\textbf{Lower bound:} 
Let $N$ be such that $(\Lambda_{\mu}^N)'
(\lambda)>R(\mu)-\frac{1}{2}\epsilon$ over a neighbourhood of $0$. As $\tilde\mu\rightarrow\mu$ we will have $(\Lambda_{\tilde\mu}^N)'(\lambda)>R(\lambda)-\epsilon$ in this neighbourhood of zero, thus $\Lambda_{\tilde{\mu}}^N(\lambda)\ge\tilde\mu_{\text{AD}}+(R(\mu)-\epsilon)\lambda$. When $R(\mu)-\epsilon>1$ this prohibits fixed points in this neighbourhood of $0$. If $R(\mu)-\epsilon< 1$, this establishes the lower bound $$\frac{\lambda_{\tilde\mu}}{\tilde\mu_{\text{AD}}}\ge \frac{1}{1-R(\mu)+\epsilon}\qquad \text{as }\,  \tilde\mu\rightarrow \mu.$$

\textbf{Finite Memory Strategies:} Suppose $\mu$ is a memory-$N$ strategy with $R(\mu)<1$. Construct a sequence of memory-$N$ strategies $\hat{\mu}_\lambda$ to be such that $\hat\mu_\lambda[\Sigma'_\sigma]=\mu[\Sigma'_\sigma]$ where $\Sigma'_\sigma$ is the set of states reachable from $\sigma$, for any $\sigma\in\supp{\mu}$. Furthermore, within each of these connected components let $\hat{\mu}_\lambda$ be the unique steady state distribution when deviations are observed at rate $\lambda$. Since it is a steady state, $\Lambda^M_{\hat\mu_\lambda}(\lambda)=\Lambda^N_{\hat\mu_\lambda}(\lambda)$ for any $M$. Note $\Lambda^N_{\hat\mu_\lambda}(\lambda)= \Lambda^N_\mu(\lambda)<\lambda$ for small $\lambda$, so there exists $\eta_\lambda>0$ such that
$$
    \lambda=\eta_\lambda+(1-\eta_\lambda)\Lambda^N_{\hat\mu_\lambda}(\lambda).\\
$$
Thus, the sequence of strategies $\tilde\mu_\lambda:=\eta_\lambda\sigma_{\text{AD}}\oplus(1-\eta_\lambda)\hat\mu_{\lambda}$ is stationary with positive mass on $\sigma_{\text{AD}}$. 

Using the first order approximation, $\Lambda^N_{\hat\mu_\lambda}(\lambda)=R(\hat\mu_\lambda)\lambda+O(\lambda^2)$, we see that the deviation rate satisfies
\begin{equation*}
    \frac{\lambda_\ast}{\eta_\lambda}=\frac{1}{1-R(\hat\mu_\lambda)}+O(\eta_\lambda).
\end{equation*}
As $\lambda\rightarrow 0$, $\hat\mu_\lambda\rightarrow\mu$, and $
R(\hat\mu_\lambda)\rightarrow R(\mu)$, obtaining the lower bound.

To show that there is a sequence of PSEs, note that in any convex neighbourhood $B\ni\mu$ there are sequences of memory-$N$ stationary strategies $\underline{\mu}_\lambda,\overline{\mu}_\lambda\rightarrow \mu$ with $D(\underline{\mu}_\lambda)<\overline{\delta}_1<D(\overline{\mu}_\lambda)$ for sufficiently small $\lambda$ and $R(\underline{\mu}_\lambda),R(\overline{\mu}_\lambda)\rightarrow R(\mu)$. By taking an appropriate convex combination of $\overline{\mu}_\lambda,\underline{\mu}_\lambda$ for each $\lambda$, we obtain a sequence of PSEs converging to $\mu$.   

Since any PSE supported on finite memory strategies is the limit of finite memory PSEs, we can use appropriate approximations of the latter PSEs to approximate the former.
\end{proof}

\begin{prop}\label{prop:rob.edge}

    If $\mu$ is a Finite Deviation PSE with reproduction rate $R(\mu)=1$, then:
    \begin{enumerate}
        \item $\mu$ is globally stable unless it is the locally unstable pure strategy $\mu=\sigma_1$.
        \item if either $\beta>1$, or $\mu$ is not concentrated, then $\mu$ is robust.
        \item if $\mu$ is concentrated and $\beta\le 1$ then $\mu$ is not the limit of positive, locally stable PSE $\tilde{\mu}$ with $\tilde{\mu}[\sigma_{\text{AD}}]>0$. In particular $\mu$ is not robust in the Finite Deviation strategy space.
    \end{enumerate}
\end{prop}
The last statement suggests that such equilibria, while possibly technically robust in larger strategy spaces, are not convincingly so: any approximating sequence either uses unstable strategies (unlike our other cases), or requires `exotic' non-positive strategies.
\begin{proof}
\textbf{Stability:} This case is covered at the end of the proof of Lemma \ref{lem:stab}.

\textbf{Non-concentrated:} If $R(\mu)=1$ and $\mu$ is not concentrated, then by taking convex combinations with the more robust concentrated equilibrium, it can be approximated with PSE $\hat{\mu}_k$ with $R(\hat{\mu}_k)<1$, each of which can be further approximated with PSE $\tilde{\mu}_{k,\ell}$ with $\tilde{\mu}_{k,\ell}[\sigma_{\text{AD}}]>0$. By taking $k\wedge \ell\rightarrow \infty$ we obtain our approximating sequence. 

    $\boldsymbol{\beta>1}$: Suppose we have a concentrated PSE $\mu$ with $R(\mu)=1$ and $\beta>1$. That is, $\mu$ is the 1-Deviation equilibrium with $\mu_{1}=\frac{1}{\beta}<1$ and $\mu_0>0$ --- this occurs when $\delta=\overline{\delta}_1$. For $\lambda>0$, suppose the strategy $\tilde\mu_{\lambda}$ puts total measure on the contributing and deviating 1-Deviation states $\tilde\mu_\lambda[\Sigma_1']$ solving \begin{equation*}
    \begin{split}
        (\Lambda^1_{\tilde{\mu}_\lambda})'(\lambda)=\tilde\mu_\lambda[\Sigma_1']\E[b\sim F_\beta]{b(1-\lambda)^{b-1}}=1,
    \end{split}
\end{equation*}
Note that this measure will be less than 1 for $\lambda$ in a neighbourhood of $0$. Putting the remaining mass on constant strategies $\sigma_{0},\sigma_{\text{AD}}$, the discounted reproduction rate will be $D(\tilde{\mu})=D(\mu)=\overline{\delta}_1$. We then set $\tilde{\mu}_\lambda[\sigma_{\text{AD}}]$ to ensure the steady state distribution has contribution rate $\lambda$, solving 
$$
\lambda=\tilde{\mu}_\lambda[\sigma_{\text{AD}}]+\tilde\mu_\lambda[\Sigma_1']\E[b\sim F_\beta]{1-(1-\lambda)^{b}}.
$$
Taking $\lambda\rightarrow 0$, we obtain our sequence of PSE $\tilde\mu_\lambda\rightarrow \mu$.

$\boldsymbol{\beta\le 1}$: The first part is an immediate consequence of the proof of Proposition \ref{prop:mono.suff}c) when $\beta\le 1$ showing that no such partial partial participation PSE exist. The second part follows since $\Lambda_{\tilde\mu}^N$ is concave for any population strategy $\tilde\mu$ on the Finite Deviation strategy space, thus if $\Lambda_{\tilde\mu}^N(0)\ge\tilde\mu[\sigma_{\text{AD}}]>0$ then the unique fixed point $\lambda_\ast$ of $\Lambda_{\tilde\mu}^N$ will have $(\Lambda_{\tilde\mu}^N)'(\lambda_\ast)\le 1-\tilde\mu[\sigma_{\text{AD}}]$ for every $N$, implying stability, and thus $\tilde\mu\not\rightarrow \mu$ by the first part.
\end{proof}
\subsection*{Section \ref{sec:strat.gen} Proofs}
\begin{proof}[Proof of Proposition \ref{prop:monotone}]
$2\Rightarrow1$: Suppose $\mu$ is impulse equivalent to a strategy $\mu'$ in a monotone strategy space with shared impulse response $\nu$. If the agent contributes $n$ periods after observing a deviation (occurs w.p. $1-\nu_n$), then they will only observe contributions, and due to the constraints of the strategy space they will then contribute in the $(n+1)$th period after the deviation. Thus $1-\nu_{n+1}\ge 1-\nu_n$, and the strategy is monotone.

$1\Rightarrow3$: Note that a strategy $\mu\in\Delta\{0,1,\dots,\infty\}$ has impulse response $\nu_n(\mu)=\mu_\infty+\sum_{k=n}^\infty \mu_n$. For decreasing sequences $(\nu_n)$ this has the inverse function $\mu_n(\nu):=\nu_n-\nu_{n+1}$ for $1\le n<\infty$, $\mu_0(\nu):={1-\nu_1}$, and $\mu_\infty(\nu):=\lim\nu_n$.

$3\Rightarrow 2$: the space of strategies $\{0,1,\dots,\infty\}$ \textit{is} monotone.
\end{proof}
\begin{proof}[Proof of Proposition \ref{prop:mono.suff}]
    \textbf{(a):} Since $D(\sigma_\infty)\ge D(\mu)=\overline{\delta}_1$, the Grim Trigger equilibrium exists.

    \textbf{(b):} Define $\hat\mu\in\Delta\mathbb{N}_0 $ to be the Finite Deviation strategy
\begin{equation}\label{eq:concentrate}
        \hat\mu_n:=\begin{cases}
        \sum_{t=1}^\infty\nu_t(\mu)-(n-1)&n\ge  \sum_{t=1}^\infty\nu_t(\mu)\ge n-1\\
        n+1-\sum_{t=1}^\infty\nu_t(\mu)&n+1\ge  \sum_{t=1}^\infty\nu_t(\mu)> n\\
        0&\text{otherwise},
    \end{cases}
\end{equation}
    essentially moving the deviations in $\mu$ as early as possible\footnote{For Finite Deviation strategies $\hat{\mu}$ is the maximal mean-preserving contraction of $\mu$.}. This is a full-participation monotone strategy with $\Re(\hat\mu)=\Re(\mu)\le 1$ (by Lemma \ref{lem:rob}) such that $\nu(\hat\mu)$ precedes $\nu(\mu)$ in FOSD, thus $D(\hat\mu)> D(\mu)$. We can then construct a monotone PSE $\mu'$ by randomizing between 0-Deviation and $\hat\mu$, which will have reproduction rate
    $$
    (1-\mu_0)\Re(\hat\mu)=(1-\mu_0)\Re(\mu)< \Re(\mu).
    $$

    \textbf{(c):} We can use the same construction $\mu'$ in the case where $\Re(\mu)<1$ or $\mu'$ is not impulse equivalent to $\mu$. Thus it remains to check the case where $\mu$ and $\mu'$ are impulse equivalent and $\Re(\mu)=1$. In this case $\mu$ is the concentrated equilibrium. We break this edge case into three cases $\beta>1$, $\beta=1$, and $\beta<1$. 
    
    If $\boldsymbol{\beta>1}$, Proposition \ref{prop:rob.edge} shows that the concentrated equilibrium is globally robust and stable. 
        
    If $\boldsymbol{\beta=1}$ and the concentrated equilibrium $\mu'$ has $R(\mu')=1$, then $$\nu(\mu')=\nu(\mu)=(1,0,0,\dots)$$and agents copy their previous observation w.p. 1 in each PSE. This cannot occur if $\overline{b}>1$, as then there is positive probability that agents observe both $C$ and $D$ in a period (since $\mu$ is assumed to be partial participation), and therefore cannot copy both of these observations the next period. Thus agents always observe precisely one action which they copy the next period, which is not locally stable --- any fluctuation in the initial deviation rate persists indefinitely.

    If $\boldsymbol{\beta<1}$, the concentrated equilibrium $\mu$ with $R(\mu)=1$ is supported on $\{n,n+1\}$ for some $n\ge1$. But for stationary partial participation strategies $\mu$, there is at most one $n$ such that $\nu_n(\mu)=1$ --- as this implies that for a.e. sequence of observations, agents copy their observation from $n$ periods ago (should they make one), this cannot hold for multiple $n$ where any finite sequence of observations occurs with positive probability (ie. unless $\lambda=0,1$). Thus $\mu$ cannot be impulse equivalent to $\mu'$.
\end{proof}
\newpage \setcounter{page}{1}
\section*{\huge Online Appendix}
\section{Generalized Preference}\label{app:gen.mod}
Let $\lambda_t$ represent the measure of deviating agents in period $t$. Note that agents cannot affect $\lambda_t$ through their influence (as they influence only a countable measure of agents). Suppose the paths $p^C,p^D$ correspond to the same sequence of $\lambda$, we suppose agents prefer the path $p^C$ to $p^D$ when
\begin{equation}
        \sum_{\tau\ge t} \Delta_{\tau-t} (p^C_{i,\tau}-p^D_{i,\tau})\le \Delta_{\lambda(\cdot)}(p^C-p^D).
    \end{equation}
Where the function $\Delta_{\lambda(\cdot)}$ captures the difference in externality between the two paths. We might consider the benchmark (or lower bound) externality given by
\begin{equation}\label{eq:Delta}
    \underline{\Delta}_{\lambda(\cdot)}(p):= \sum_{(j,\tau)\in\mathcal{C}^+_{i,t}} \underline{\delta}^{\tau-t}\underline{\alpha}_{\lambda_\tau}p_{j,\tau},
\end{equation}
where $\lambda\mapsto \underline{\alpha}_{\lambda}$ is continuous and positive\footnote{Changing how $\alpha$ varies with $\lambda$ can affect our edge case analysis in Proposition \ref{prop:rob.edge}. Moreover Proposition \ref{prop:mono.suff}c) may no longer hold if $\alpha$ is an increasing function of the deviation rate $\lambda$ (ie. contributing is submodular, not complementary), in which case increasing the deviation rate $\lambda$ effectively decreases $\overline{\delta}_1$, potentially allowing new equilibria to exist at lower participation rates that may not exist at higher rates.}, and $\mathcal{C}^+_{i,t}$ is the `cone' of actions that can be influenced by agent $i$'s action at time $t$ (formally defined in Appendix \ref{app:obs.ind}). The body of the paper uses this externality with $\underline{\alpha}_{\tau}\equiv \alpha$ constant. 

The key properties necessary for our result are that $\Delta_{\lambda(\cdot)}$ is present biased, continuous, and bounded below by an expression like eq. \ref{eq:Delta}. This describes a wide range of preferences, to demonstrate the range covered by these conditions, consider the following example:

\begin{ex}[Local Externalities on the Plane]
    Agents are indexed by $i\in\mathbb{I}:=\mathbb{R}^2$, each period they face a decision to use a clean technology $C$ or a polluting technology $D$. The polluting technology is more cost-efficient by an amount $c>0$, but causes an externality on agents that decreases with the distance between agents. Specifically, the externality imposed on agents a distance $r$ away is given by $f(r)$ where $f:\mathbb{R}_+\rightarrow \mathbb{R}_{++}$ is a decreasing function bounded: $c> f(r)\ge a_0\exp(-kr^2)$ for some $a_0,k\in\mathbb{R}_{++}$.
    
    The number of agents that observe an action is drawn from $\text{Poisson}(\beta)$, and the location of these agents is independently drawn from a normal distribution centred on $i$: $N(i,\sigma_o^2)$.

    The agent $i$ at time $t$ prefers the path $p^C$ to $p^D$ when
    \begin{equation}\label{eq:pollute}
        \sum_{t\in\mathbb{T}} \delta^{t} \Big[-c(p^C_{i,t}-p^D_{i,t})+\sum_{j\in\mathbb{I}} f(\|j-i\|)(p^C_{j,t}-p^D_{j,t})\Big]\ge 0.
\end{equation}
The location of an agent $D$ `degrees of separation' away can then be calculated through convolutions of normal distributions to be distributed according $N(0,D\sigma_0^2)$. and thus the expected externality from such an agent is bounded below by
$$
\Econd{f(r)}{D}\ge \int a_0 e^{-kr^2}e^{-r^2/(2D\sigma_0^2)}\frac{2\pi r\,dr}{2\pi D\sigma_o^2}
=\frac{a_0}{2kD^2\sigma_o^4+1}.
$$
Since degree of separation of an agent affected by agent $i$'s action $n$ periods later is at most $n$, the coefficient on agent $j$'s defection $n$ periods later is bounded (as a function of $n$)
$$
\tfrac{a_0\delta^n}{2kn^2\sigma_o^4+1}\ge a\underline{\delta}^n
$$
for some choice of $a>0$ for every $\underline{\delta}<\delta$. Comparing with $\underline{\Delta}$ defined in eq. \ref{eq:Delta} (and normalizing $c=1$), we see that it acts as a lower bound on the externality in this model when we define in this example is bounded below  a lower bound is obtained with $\underline{\alpha}=\frac{a}{c}$.
\end{ex}

Our results in this paper rely following four properties satisfied by the above example
\begin{enumerate}
    \item \textbf{Quasi-discounting}: there exist $\underline{\delta}>0,\underline{\alpha}:[0,1]\rightarrow\mathbb{R}_{++}$ such that eq. \ref{eq:Delta} is a lower bound on $\Delta_{\lambda(\cdot)}$.
    \item \textbf{Necessity of influence}: if $p^C_{j,\tau}=p^D_{j,\tau}$ for all $j\neq i$ and $\tau$ then eq. \ref{eq:pref} fails for the constant path $\lambda(\cdot)\equiv 0$.
\end{enumerate}
This means that if agents knew that they had no influence on others --- perhaps their action is private --- then deviating from the norm would be a dominant action. This implies $\underline{\alpha}_0<1$.

Our last properties requires the degree of separation $d_\mathcal{C}:\mathbb{I}^2\rightarrow\mathbb{Z}_+$  between two agents in the same cone of influence $\mathcal{C}$ that captures the number of observations that separate two agents:
\begin{equation*}
    d_\mathcal{C}(j,i):=\begin{cases}
        0&j =i\\
        1+\min_{i'}\{d_\mathcal{C}(i',i); j\in O(i',\tau)\text{ for some $\tau$}\}& j\neq i
    \end{cases}
\end{equation*}
We say that a path $p\in\mathcal{P}_{i,t}$ \textbf{first order dominates} another path $p'$, denoted $p\trianglerighteq p'$, if at any distance $m$, it contains on average more deviations before any given time $T$, formally:
\begin{equation}
    \label{eq:present-bias}
    \sum_{\substack{d_\mathcal{C}(j,i)=m,\\ \tau<T}} (p_{j,\tau}-p'_{j,\tau})\ge 0\qquad\text{for all }m,T
\end{equation}
If $p \trianglerighteq p'$ then it can be obtained from $p'$ by a combination of moving deviations earlier in time (at the same radius $m$) and increasing the likelihood of deviating for various agents $j$ in periods $\tau$. It is natural that these transformations should increase the size of the externality to an agent:
\begin{enumerate}\setcounter{enumi}{2}
    \item \textbf{Monotonicity}/\textbf{Present-bias}: $p\mapsto \Delta_0(p)$ is $\trianglerighteq$-increasing.
\end{enumerate}
The last property concerns the continuity of the function $\Delta$ over the space of paths. There are two notions of continuity that we will use.

Concerning the existence of equilibria in Section \ref{sec:exist}, it suffices to have continuity of $p\mapsto \Delta_0(p)$ along one-dimensional lines $\mu_0\mapsto p^{\mu_0}$ through $\mathcal{P}$ given by
    \begin{equation*}
        p_{j,\tau}^{\mu_0}:={\mu_0}^{d_\mathcal{C}(j,i)}p_{j,\tau}.
    \end{equation*}
    This transformation corresponds to agents playing a constant strategy with probability ${\mu_0}$ and the mixed strategy corresponding to $p$ with probability $1-{\mu_0}$ (each agent in effect `blocks' agent $i$'s influence w.p. ${\mu_0}$). This moderates the influence an agent has on those further away in their causal cone.
    
When considering robust equilibria in Section \ref{sec:robust}, we constrain attention to a normed subspace $\ell^1_{\delta_\ast}\subseteq\mathcal{P}$, where $\delta_\ast\in]0,1[$ is a specific parameter. This space is defined
\begin{align*}
    \|p\|_{\delta_\ast}:=&\sum_{j,\tau}{\delta_\ast}^{d_\mathcal{C}(j,i)}|p_{j,\tau}|&
    \ell^1_{\delta_\ast}:=&\big\{p\in\mathcal{P};\|p\|_{\delta_\ast}<\infty\big\}.
\end{align*}
This norm discounts the actions of agents that are further away at rate $\delta_\ast$.
    
Formally, our two continuity conditions are:
\begin{enumerate}\setcounter{enumi}{3}
    \item \textbf{Regularity}: 
    \begin{enumerate}
        \item for all $p\in\mathcal{P}$, ${\mu_0}\mapsto \Delta_0(p^{\mu_0})$ is continuous
        \item for some $\delta_\ast<1$, $(\lambda,p)\mapsto \Delta_\lambda(p)$ (where $\lambda$ indicates a path with constant cooperation rate $\lambda$) and $\lambda\mapsto \alpha_\lambda$ are continuous
    \end{enumerate}
\end{enumerate}
Note that $\mu_0\mapsto p^{\mu_0}$ is a continuous map under $\|\cdot\|_{\delta_\ast}$ when $p^\mu\in \ell^1_{\delta_\ast}$; the added value of (4a) is to extend continuity to $p\not\in\ell^1_{\delta_\ast}$.
\subsection{Proofs of Preference Properties}
\begin{lem}\label{lem:exist}
    If a reactivity distribution $\mu^+$ corresponds to a full participation PSE when the externality is given by eq. \ref{eq:Delta}, then it also corresponds to a PSE under assumptions (1)-(4a).
\end{lem}
\begin{proof}
    By assumption (1), we have that the externality under the strategy $\mu^+$ is greater than $1$. By assumption (2), we know that under the pure strategy of 0-Deviation the externality is less than $1$. By assumption (4a), the externality is continuous as we increase the weight $\mu_0$ placed on 0-Deviation strategies. Intermediate value theorem then guarantees that there exists a weight $\mu_0$ with which agents can mix between $\mu_0$ and $\mu^+$ to obtain equilibrium.
\end{proof}
\begin{lem}\label{lem:SOSD}
    Let $\tilde{\mu},\mu\in\Delta\Sigma$ be two strategies, to verify if their paths after observing a deviation satisfy $p'(\tilde{\mu})\trianglerighteq p'(\mu)$ it suffices to check eq. \ref{eq:present-bias} for $m=1$.
\end{lem}
\begin{proof}
    Let
    $$
    F_m^{\mu}(T):=\sum_{\substack{d_\mathcal{C}(j,i)=m,\\ \tau<T}} p'_{j,\tau}(\mu)
    $$ 
    be the expected number of deviations an agent a distance $m$ from $(i,0)$ will make before time $T+m$, where $T\ge 1$ (note that it take at least $m+1$ periods for these agents to enter the causal cone $\mathcal{C}_{i,0}^+$).

    Note that these functions can be decomposed: the expected number of deviations at a distance of $m+1$ and time $T$ is the expected number of deviations at a distance $m$ and time $T-\tau$ times the expected number of deviations at distance 1 from a deviating agent $\tau$ periods after the deviation:
    \begin{equation*}
        \begin{split}
            F_{m+1}^\mu(T)=&\sum_{\tau=1}^{T-1} F_1^{\mu}(T-\tau)\, \beta (F_m^{\mu}(\tau)-F_m^{\mu}(\tau-1))
        \end{split}
    \end{equation*}
    or equivalently, using convolution notation,
    \begin{equation*}
        \begin{split}
            dF_{m+1}^{\mu}=&\beta dF_1^\mu\ast dF_m^\mu.
        \end{split}
    \end{equation*}
    Thus if $F_1^{\mu}\ge F_1^{\tilde\mu}\ge 0$, then inductively $F_m^{\mu}\ge F_m^{\tilde\mu}\ge 0$.
\end{proof}
We describe some results for Finite Deviation strategies, as well as for more general strategies.
\begin{cor}\label{cor:SOSD}
    For general strategies, $\mu,\tilde{\mu}\in\Delta\Sigma$, $p'(\mu)\trianglerighteq p'(\tilde\mu)$ iff their corresponding impulse responses satisfy $\sum_{\tau\le T}\nu_\tau(\mu)\ge \sum_{\tau\le T}\nu_\tau(\tilde\mu)$ for all $T$.
    
    For Finite Deviation strategies $\mu,\tilde\mu\in\Delta\mathbb{N}_0$, $p'(\mu)\trianglerighteq p'(\tilde\mu)$ if $\mu$ is SOSD larger than $\tilde{\mu}$.
\end{cor}
This allows us to extend Prop. \ref{prop:exist}, Corollary \ref{cor:SOSD}, and ensuring that the concentrated equilibrium is the PSE that minimizes the reproduction rate.
\begin{proof}
    For general strategies $\mu\in\Delta \Sigma$
    $$
    F_1^\mu(T)=\beta\sum_{\tau\le T}\nu_\tau(\mu).
    $$
    For Finite Deviation strategies $\mu\in\Delta\mathbb{N}_0$, the corresponding quantity is an expectation of an increasing and concave function of $n$:
    \begin{equation*}
        F_1^\mu(T)=\beta\E[n\sim \mu]{\max\{n, T\}}.\qedhere
    \end{equation*}
\end{proof}
\begin{lem}[Continuity]\label{lem:cont}
    Suppose $\beta\|\nu(\mu)\|_1<1$, and $\nu(\hat\mu)\rightarrow \nu(\mu)$ in $\ell^1$, then $\Delta_\lambda(p(\hat\mu))\rightarrow \Delta_\lambda(p(\mu))$ under assumption (4b).
\end{lem}
This continuity property is useful for extending Lemma \ref{lem:rob} for bounded memory strategies.
\begin{proof}
Let $f^\mu_m(t)$ be the expected number of deviating agents at distance $m$ from the deviating agent $t+m$ periods after the deviation. These probability mass functions have the convolution propery 
\begin{equation*}
    f_{m+1}^\mu=f_1^\mu\ast f_m^\mu.
\end{equation*}
Where $f_1^\mu(t)=\beta\nu_t(\mu)$. By induction we have $\|f_m^\mu\|_\infty\le \|f_m^\mu\|_1<1$ when $\beta\|\nu(\mu)\|_1<1$.

Suppose, then that $\beta\|\nu(\hat\mu)\|_1<1$, then inductively
\begin{equation*}\begin{split}
    \big\|f_{m+1}^\mu-f_{m+1}^{\hat\mu}\big\|_1\le& \big\|f_1^\mu\big\|_\infty\big\|f_m^{\mu}-f_m^{\hat\mu}\big\|_1+\big\|f_m^\mu\big\|_\infty\big\|f_1^\mu-f_1^{\hat\mu}\big\|_1\\
    \le& \big\|f_m^\mu-f_m^{\hat\mu}\big\|_1+\big\|f_1^\mu-f_1^{\hat\mu}\big\|_1\le (m+1)\big\|f_1^\mu-f_1^{\hat\mu}\big\|_1
    \end{split}
\end{equation*}
We can then bound the distance between the paths $p(\hat\mu),p(\mu)$:
\begin{equation*}
    \|p(\hat\mu)-p(\mu)\|_{\delta_\ast}=\sum_{m=1}^\infty\delta_\ast^m\big\|f_m^\mu-f_m^{\hat\mu}\big\|_1\le \beta\|\nu(\mu)-\nu(\hat\mu)\|_1\sum_{m=1}^\infty\delta_\ast^m(m+1)
\end{equation*}
where the last summation is finite by, e.g., ratio test.
\end{proof}
\end{document}